\documentclass[11pt]{article}
\usepackage[utf8]{inputenc}
\usepackage{setspace}
\usepackage{graphicx}
\graphicspath{{images/}{../images/}}
\usepackage[margin = 1in]{geometry}
\usepackage{graphicx}
\usepackage{float}
\usepackage{blindtext}
\usepackage{amsmath,amssymb}
\usepackage{amsfonts}
\usepackage{multirow}
\usepackage[linesnumbered, ruled, vlined]{algorithm2e}
\SetKwInput{KwInput}{Input}
\SetKwInput{KwInitialise}{Initialise}
\SetKwInput{KwOutput}{Output}
\SetKwInput{KwRepeat}{Repeat}
\DeclareMathOperator*{\argmax}{arg\,max}
\DeclareMathOperator*{\argmin}{arg\,min}
\usepackage{graphicx}
\usepackage{float}
\usepackage{caption}
\usepackage{subcaption}
\usepackage{bbm}
\usepackage{parskip}
\usepackage{booktabs}
\usepackage{hyperref}
\usepackage[nottoc,notlot,notlof]{tocbibind}
\usepackage{color,soul}
\usepackage[super]{nth}
\usepackage{amsthm}
\usepackage{mathtools}
\usepackage{enumitem}

\let\oldnl\nl
\newcommand{\nonl}{\renewcommand{\nl}{\let\nl\oldnl}}

\newtheorem{theorem}{Theorem}[]
\newtheorem{corollary}{Corollary}[]

\newtheorem{proposition}{Proposition}[]
\newtheorem{lemma}{Lemma}[]
\newtheorem{definition}{Definition}[]

\usepackage{titlesec}
\usepackage{siunitx}

\usepackage[most]{tcolorbox}

\tcbset{
    frame code={}
    center title,
    left=0pt,
    right=0pt,
    top=0pt,
    bottom=0pt,
    colback=gray!70,
    colframe=white,
    width=\dimexpr\textwidth\relax,
    enlarge left by=0mm,
    boxsep=5pt,
    arc=0pt,outer arc=0pt,
    }
    
\usepackage[isbn=false, citestyle = authoryear, backend=biber,url=false, date=year, style = apa, doi=false,maxcitenames=1]{biblatex}

\AtEveryBibitem{\clearfield{month}}
\AtEveryBibitem{\clearfield{day}}
\DeclareAutoCiteCommand{inlinetext}{\textcite}{\textcites}
\usepackage{multicol}
\usepackage{adjustbox}
\usepackage{rotating}
\usepackage{booktabs}
\usepackage{hvfloat}

\begin{document}

\title{A stochastic game framework for patrolling a border}
\date{}
\author{Matthew Darlington \thanks{STOR-i Centre for Doctoral Training, Lancaster University, Lancaster LA1 4YR, M.Darlington@lancaster.ac.uk}, Kevin D. Glazebrook \thanks{Department of Management Science, Lancaster University, Lancaster LA1 4YX, K.Glazebrook@lancaster.ac.uk}, David S. Leslie \thanks{Department of Mathematics and Statistics, Lancaster University, Lancaster, LA1 4YR, D.Leslie@lancaster.ac.uk}, Rob Shone \thanks{Department of Management Science, Lancaster University, Lancaster LA1 4YX, R.Shone@lancaster.ac.uk}, Roberto Szechtman \thanks{Operations Research Department, Naval Postgraduate School, Monterey, California 93943, rszechtm@nps.edu}}

\maketitle

\begin{abstract}
In this paper we consider a stochastic game for modelling the interactions between smugglers and a patroller along a border. The problem we examine involves a group of cooperating smugglers making regular attempts to bring small amounts of illicit goods across a border. A single patroller has the goal of preventing the smugglers from doing so, but must pay a cost to travel from one location to another. We model the problem as a two-player stochastic game and look to find the Nash equilibrium to gain insight to real world problems. Our framework extends the literature by assuming that the smugglers choose a continuous quantity of contraband, complicating the analysis of the game. We discuss a number of properties of Nash equilibria, including the aggregation of smugglers, the discount factors of the players, and the equivalence to a zero-sum game. Additionally, we present algorithms to find Nash equilibria that are more computationally efficient than existing methods. We also consider certain assumptions on the parameters of the model that give interesting equilibrium strategies for the players. 
\end{abstract}

\section{Introduction}

Ranging from drug trafficking across the U.S.-Mexico border (\cite{Gutierrez2021}), to oil smuggling out of Nigeria (\cite{Ojewale2021}), and illegal fishing in the continental shelf off South America (\cite{Goodman2021}), the problem of how to patrol a border is fundamental to government organisations worldwide. How to patrol well is a challenging problem because it is infeasible to protect everywhere simultaneously due to constraints on resources, and thus a carefully thought out strategy is required. The associated trade-off is a complex one: if the patrols are too predictable the smugglers may be able to easily figure out where and when they can get through undetected. However, if the patrollers act too randomly they may not be adequately protecting the most vulnerable sections. In this work we introduce a stochastic game model for patrolling a border, detail how the strategies for both the patroller and smugglers can be found, and then analyse the solutions obtained.

Specifically, we consider a scenario where a single patroller attempts to stop a group of cooperating smugglers taking items across a border. The border here is thought of as being a finite set of locations which could be roads, border control posts or even an area of air, land or sea. The smugglers attempt to send some illicit items through these locations, and it is the patroller's goal to find an efficient strategy for stopping these items from getting through. It is assumed that there are known and fixed rewards and penalties that the smugglers receive or incur if they are respectively successful or not, proportional to the quantity of items they attempt to smuggle. Similarly, the patroller receives a reward or penalty depending on whether they stop the smugglers or not. There is a single smuggler fixed at each location. However, the patroller must traverse the geography of the border and pay a cost to do so. The patroller and smugglers make these decisions through time, needing to account for both their immediate reward and how their future rewards will be affected.

There is a significant operations research literature on patrol problems that focusses on modelling real-world situations. Examples include \cite{Urrutia2000} looking at the protection of galleries containing expensive paintings, or \cite{Richard1972} considering the daily patrol patterns of a police officer in the United States. An example where developed methods have been implemented in practice concerns the protection of the Los Angeles International Airport and is by \cite{Pita2009}. Pita et al.'s work models the problem as a Bayesian Stackelberg game to give the patrollers a randomised method to protect the airport from threats. The authors reported ``very positive feedback about the deployment''. 

We consider a game theoretic approach to the problem of patrolling a border where the patroller and the smugglers take actions simultaneuously. This was first considered by \cite{Alpern2011} who look at protecting against a single attempt by a smuggler that takes a fixed time to complete. An important assumption made in the paper is that the outcome of the attempt results in a win or loss for the patroller. \cite{Lin2013} consider a similar model but introduce the possibility of the attempt taking a non-deterministic length of time to complete. \cite{Lin2014} further advance this work by considering the situation where the patroller has a chance to miss the attempt taking place. Another extension of this model is by \cite{McGrath2017} who solve a problem in which there is a non-trivial difference in both the time taken to travel around the locations and to check if an attempt is in progress at each location. The application of patrolling a border is considered by \cite{Papadaki2016} and \cite{Alpern2019}, but it is still based on the assumption that the adversary makes only one attempt. Recent work in the patrolling literature includes \cite{Alpern2021} who consider a problem in which the patroller chooses whether or not to wear a uniform and \cite{Lin2021} who considers how to optimally patrol the perimeter of a location. 

There are two key assumptions made by these papers that are not consistent with the problem of patrolling a border we consider. Firstly, in our setting a single successful smuggler is not catastrophic to the patroller. Instead, we have a small penalty incurred by the patroller that depends on the amount of items that are trafficked. Secondly, the normal-form game approaches discussed in previous papers can only consider trying to stop one attempt without taking into consideration what happens next. This means that while we might catch the smugglers once, the patroller could be left exposed for an upcoming series of attempts. This motivates us to develop a new stochastic game model for patrolling a border. We will justify its benefits empirically with numerical experiments.

There are two bodies of work in the literature which look at a similar problem to ours. \cite{Grant2020} examine patrolling a border against opponents who make many small attempts to smuggle items. However, the smugglers are assumed to be acting at random by the authors whereas we make the stronger assumption that the smugglers pick actions strategically. The model closest to ours is discussed in \cite{Filar1986} and \cite{Filar1985}, but is focused on a different problem setting. They consider a travelling inspector who checks factories to detect the illegal dumping of materials. The model differs to the one considered here in having only a finite number of possible actions available to the inspector's adversaries. We consider a setup in which the actions available to the smugglers lie in a continuous set. The change from finite to continuous action set significantly complicates the task of solving the game. We provide a solution which meets this challenge.

From the modelling perspective, our main contribution is the extension of the smugglers' action set to a continuous interval. Having this larger set of alternatives gives a more realistic formulation, especially so in the case in drug, oil, and weapons smuggling scenarios, where the smugglers' action set is essentially a continuum. Furthermore, we propose a stochastic game rather than a normal-form game for our problem and so the patroller must consider the trade-off between immediate and future rewards. The increased resolution of our model adds additional complexity to the challenge of finding the Nash equilibria for the patroller and the smugglers. This prompts the main methodological contributions of our work. These begin with the elucidation of the structure in the Nash equilbria of the game and lead to the development of solution algorithms to overcome the resulting computational issues. We prove that these algorithms find or converge to the optimal solution and, moreover, that they are computationally faster than existing methods in cases where a comparison is meaningful.

The rest of the paper is organised as follows. In Section 2, we present our stochastic game framework for patrolling a border. Section 3 gives an overview of Nash equilibria that arise in the model, and establishes several of their properties. Section 4 provides an analysis of the methods to find Nash equilibria in the border patrol game. In Section 5, we make additional assumptions on the cost function which leads to more detailed characterisations of equilibria. An empirical analysis of the performance of our approach is given in Section 6, along with a discussion of the solutions to specific instances of our model. Section 7 concludes our paper with a summary and suggestions for future work. All proofs not included in the main body of the paper can be found in an Appendix.

\section{Model Description}

We consider a border made up of $n$ locations labelled from 1 to $n$ inclusive. In the paper we will use the notation $[n]$ to denote the set of all locations where $[n] = \{1, \dots, n\}$. Time will be modelled in discrete steps $t = 0, 1, \dots$. Such time steps are natural here , where decisions could be taken on an hourly or daily basis. 

We present the model in this section by taking the smugglers collectively to be a single player. This will be the case throughout the paper unless explicitly noted otherwise. Thus we look to define a stochastic game between two players: a single patroller and the smugglers. The patroller begins each time step $t$ at some location $s_t$, which we take to be the current state of the system. Hence, the state space of the game is $\mathcal{S} = [n]$. The patroller picks a location to defend, $b_t$, and the smugglers pick a quantity of items from the interval $[0,1]$ to send to each location. We write the smugglers' action as $\boldsymbol{a}_t = (a^1_t, \dots a^n_t)$, where $a_t^i$ is the quantity sent to the location $i$. Note that without loss of generality this accounts for quantities from the interval $[0,q]$ for some $q > 0$ by a scaling of the actions. Hence, the action space of the patroller and smugglers respectively at each epoch are $\mathcal{A}_{pat} = [n]$ and $\mathcal{A}_{smug} = [0,1]^n$. Both the patroller and the smugglers take an action simultaneously, with no knowledge of the action chosen by the opponent. The state of the system at the next time step is the previous action of the patroller, and so
\begin{equation}
    \mathbb{P}(s_{t+1} = b \; | \; b_t = b, \boldsymbol{a}_t = \boldsymbol{a}, s_t = s) = \mathbb{P}(s_{t+1} = b \; | \; b_t = b) = 1 \label{transition}
\end{equation} 
for all $b \in \mathcal{A}_{pat}$, $\boldsymbol{a} \in [0,1]^n$ and $s \in \mathcal{S}$. As we will see, the players can hoose their actions according to some probability distribution which results in a random state transition in the game.

The patroller catches all items sent by the smugglers to the location they have chosen to defend. At every other location, the items are successfully smuggled. Smugglers receive a fixed reward of $r_i > 0$ for each unit of item smuggled through the location $i$. However, if caught, the smugglers must pay a penalty related to the amount smuggled. This is determined by the cost function $C: [0,1] \to \mathbb{R}_{+}$. We assume that $C$ is an increasing function with $C(0) = 0$. The patroller's payoff is equal to the negative of the smugglers' payoff, but she must additionally pay a cost for moving from one location to another. These movement costs are given by the parameters $m_{i,j} \geq 0 \; i,j \in [n]$. Thus, the reward functions of the patroller and the smugglers respectively are as follows:
\begin{align*}
    R_{pat}(b, \boldsymbol{a}, s) &= C(a_b) - \sum_{i \in [n] \setminus \{b\}} r_i a_i - m_{s,b} \\
    R_{smug}(b, \boldsymbol{a}) &= \sum_{i \in [n] \setminus \{b\}} r_i a_i - C(a_b) .
\end{align*}
The game continues for an infinite number of time steps, with rewards discounted at a rate of $\gamma \in [0,1)$ for the patroller. Although smugglers can each have an individual discount rate of $\lambda_i \in [0,1)$, we prove in Section 3 that the assumption that smugglers have a discount rate $\gamma$ is without loss of generality. 

A pure action is an action which a player is able to perform. In our case these are the elements of the sets $\mathcal{A}_{pat}$ and $\mathcal{A}_{smug}$ for the patroller and smugglers respectively. Instead of picking a pure action deterministically, players can draw an action according to a probability distribution over their pure actions, which may depend on the current state of the system $s$. A stationary mixed strategy for either player is a $|\mathcal{S}|$-tuple of probability distributions over the pure actions of a player,
\begin{gather*}
    \boldsymbol{\Pi} = (\boldsymbol{\pi}^1, \dots, \boldsymbol{\pi}^n) \in (\Delta([n]))^n \\
    \boldsymbol{\Xi} = (\boldsymbol{\xi}^1, \dots, \boldsymbol{\xi}^n) \in (\Delta([0,1])^n)^n
\end{gather*}
where $\Delta(S)$ denotes the set of probability distributions with set $S$ as their support and $S^k$ is the $k$-ary Cartesian power of $S$ for a natural number $k$. The results in subsequent sections will establish that it is sufficient to consider only stationary strategies, rather than strategies with a dependence on the current time step. The strategies for the patroller and smugglers respectively given that the state of the system is $i$ are $\boldsymbol{\pi}^i$ and $\boldsymbol{\xi}^i$. Assuming the strategies are fixed over time, we write the expected discounted reward for both players over an infinite horizon as
\begin{equation*}
    U_{pat}(\boldsymbol{\Pi}, \boldsymbol{\Xi}) = \mathbb{E}_{\boldsymbol{\Pi}, \boldsymbol{\Xi}, \mathbb{P}_0} \left[ \sum_{t = 0}^\infty \gamma^t R_{pat}(b_t, \boldsymbol{a}_t, s_t) \right] ,
\end{equation*}
and, 
\begin{equation}
    U_{smug}(\boldsymbol{\Pi}, \boldsymbol{\Xi}) = \mathbb{E}_{\boldsymbol{\Pi}, \boldsymbol{\Xi}, \mathbb{P}_0} \left[ \sum_{t = 0}^\infty \gamma^t R_{smug}(b_t, \boldsymbol{a}_t) \right]. \label{U}
\end{equation}
In (\ref{U}) where expectations are taken with respect to the strategies of both players so that $b_t \sim \boldsymbol{\pi}^{s_t}$ and $\boldsymbol{a}_t \sim \boldsymbol{\xi}^{s_t}$, and also with respect to the probability distribution $\mathbb{P}_0$ over the initial state $s_0$. Since the outcome of one player depends on the action of the other, it is not possible to maximise the rewards of the players independently. We give the definition of a Nash equilibria as first given by \cite{Nash1950}.
\begin{definition}
The strategies $\boldsymbol{\Pi}^*$ and $\boldsymbol{\Xi}^*$ for the patroller and smugglers respectively form a Nash equilibrium for the game if and only if,
\begin{gather*}
    U_{pat}(\boldsymbol{\Pi}^*, \boldsymbol{\Xi}^*) \geq U_{pat}(\boldsymbol{\Pi}, \boldsymbol{\Xi}^*) \; \forall \; \boldsymbol{\Pi} \in (\Delta([n]))^n \\
    U_{smug}(\boldsymbol{\Pi}^*, \boldsymbol{\Xi}^*) \geq U_{smug}(\boldsymbol{\Pi}^*, \boldsymbol{\Xi}) \; \forall \; \boldsymbol{\Xi} \in (\Delta([0,1]^n))^n .
\end{gather*}
\end{definition}

Nash equilibria give the most natural solution for our model, in that they provide the best possible lower bound of the discounted expected reward to the patroller. This could be operationally important if, for example, the smugglers were to discover the strategy of the patroller and were able to optimise their strategy using this knowledge.

\section{Properties of Nash Equilibria}

We now seek to prove properties of the Nash equilibria in our model, which can help us to understand the behaviour of the patroller and the smugglers. Firstly, we note that the model described in the previous section falls into a class of stochastic games called single controller stochastic games.
\begin{definition}
Suppose we have an $n$-player stochastic game with players $1, \dots, n$, with player $i$ taking the action $a^i$ from action set $\mathcal{A}^i$. Then the game is a single controller stochastic game with player $j$ as the controller if and only if,
\begin{equation*}
    \mathbb{P}(s_{t+1} = s' \; | \; s_t = s, a_t^1 = a^1, \dots, a_t^n = a^n) = 
    \mathbb{P}(s_{t+1} = s' \; | \; s_t = s, a_t^j = a^j)
\end{equation*}
for all $s, s' \in \mathcal{S}$ and $a^i \in \mathcal{A}^i$ for all players $i$.
\cite{Filar1997}
\end{definition}
It follows from Equation (\ref{transition}) that our model is a single controller stochastic game with the patroller as the controller. The single controller property leads to three results about Nash equilibria in our game: Lemma 1 proving that the smugglers can be assumed without loss of generality to be a single player, Proposition 1 showing discount rates of all players can be assumed to be equal without loss of generality, and Proposition 2 giving a zero-sum formulation of the game with equivalent Nash equilibria.

\subsection{Aggregation of Smugglers}

We first show that, without loss of generality, we can assume that the smugglers act as a single cooperating player. If the smugglers were acting independently of one another, then we would have an $n+1$ player game where each smuggler has a reward function equal to,
\begin{equation*}
    R_{smug}^i(b, \boldsymbol{a}) = \begin{dcases}
    r_i a_i &\text{ if } b \neq i \\
    -C(a_i) &\text{ if } b = i
    \end{dcases}
\end{equation*}
for $i \in [n]$. The action space of smuggler $i$, $\mathcal{A}_{smug}^i$, is equal to the unit interval and his strategy, $\boldsymbol{\Xi}_i$, is in $\Delta([0,1])$. Therefore, a set of strategies for every player in this nonaggregated game is denoted $(\boldsymbol{\Pi}, \boldsymbol{\Xi}_1, \dots, \boldsymbol{\Xi}_n)$. The patroller's reward function and action space remain the same, as do the state transitions and discount factors. 

We can define a mapping from the reward function in the nonaggregated game to the aggregated game by,
\begin{equation*}
    \sum_{i=1}^n R_{smug}^i(b, \boldsymbol{a}) \to R_{smug}(b, \boldsymbol{a})
\end{equation*}
and a mapping from strategies in the nonaggregated game to the aggregated game by,
\begin{equation*}
    \left( \boldsymbol{\Pi}, \bigtimes_{i=1}^n \boldsymbol{\Xi}_i \right) \to (\boldsymbol{\Pi}, \boldsymbol{\Xi}) .
\end{equation*}
\begin{lemma}
Nash equilibria in the aggregated game coincide with those in the nonaggregated game in that if we have a Nash equilibrium in one game and map the strategies to the other, then it remains a Nash equilibrium.
\end{lemma}
A similar statement can be found by \cite{Filar1985}, and an analytical proof is presented for their travelling inspector problem. The intuition behind the proof is that since the individual smugglers have independent reward functions and since their actions make no difference to the state transitions, neither combining nor splitting the smugglers create an incentive to deviate. A consequence of Lemma 1 is that we can choose whether to analyse the strategy for a single smuggler or the aggregated group, depending on which is more trackable in the context.

\subsection{Discount Rates}

We now move on to discuss the effect of players having different discount rates in our game for patrolling a border. We prove that if every player has an individual discount rate, then Nash equilibria are equivalent to those which occur when all players have the same discount rate as the patroller. 
\begin{proposition}
Suppose that $(\boldsymbol{\Pi}^*, \boldsymbol{\Xi}_1^*, \dots,  \boldsymbol{\Xi}_n^*)$ is a Nash equilibrium for the nonaggregated game in which all players have discount rate $\gamma \in [0,1)$. It remains a Nash equilibrium in any nonaggregated game in which the patroller has a discount rate of $\gamma$, and smuggler $j$ has a discount rate of $\lambda_j \in [0,1)$, $1 \leq j \leq n$.
\end{proposition}
The consequence of Proposition 1 is that we can assume without loss of generality that all players can be assumed to have a discount rate $\gamma$.
\subsection{Zero-sum formulation of model}
A two-player zero-sum stochastic game is defined as follows.
\begin{definition}
A two-player stochastic game is zero-sum if the reward to one player is always equal to the negative of the reward to the other player.
\end{definition}
Whilst the model introduced in the previous section is not zero-sum, it only differs from one by the inclusion of the cost the patroller must pay to move around the locations. We show that if the game is modified such that the smugglers are assumed to earn a reward equal to the cost of the movement of the patroller, then the Nash equilibria of the game are unchanged. The version of the game where the smugglers get this reward is clearly zero-sum.
\begin{proposition}
    Consider a stochastic game identical to the one introduced in the previous section, but where the reward function for the patroller and smugglers respectively are
    \begin{align*}
    \tilde{R}_{pat}(b, \boldsymbol{a}, s) &= R_{pat}(b, \boldsymbol{a}, s) ,\\
    \tilde{R}_{smug}(b, \boldsymbol{a}, s) &= -\tilde{R}_{pat}(b, \boldsymbol{a}, s) = R_{smug}(b, \boldsymbol{a}) + m_{s,b} = \sum_{i \in [n] \setminus \{b\}} r_i a_i + m_{s,b} -  C(a_b).
    \end{align*}
    The new game is a two player, zero-sum stochastic game. Furthermore, the Nash equilibria for the two games are identical.
\end{proposition}    
Please note that we shall consider this altered form of the game with reward functions $\tilde{R}$ for the remainder of the paper. Proposition 2 is important since it allows us to apply a number of algorithms to find Nash equilibria which require that the game be zero-sum. Examples of these include algorithms for finite two-player zero-sum stochastic games developed by \cite{Shapley1953} and those developed for single controller games by \cite{Raghavan2003}. A further consequence of Proposition 2 is that in our model there must indeed be a Nash equilibrium with stationary strategies. This follows from the result of \cite{Maitra1970}, since in our game the reward to either player is continuous in the actions of both players and the state transition is deterministic.

However, the assumption of finite action spaces, made by \cite{Shapley1953} and \cite{Raghavan2003} does not hold in our model. The smugglers can take any action from the $n$-dimensional unit cube. Determining Nash equilibria remains a major challenge. This is the subject of Section 4.
\section{Finding Nash Equilibria}
There exist in the literature algorithms that can calculate Nash equilibria in two-player zero-sum stochastic games such as the one by \cite{Shapley1953}. However, their assumption that the game is finite means that they are not directly applicable here. In this section, we present a method for determining Nash equilibria in our game. We begin by defining the value of a state $s$ for the players.
\begin{definition}
The value of a state $s$ to the patroller, $\boldsymbol{V}_{pat}(s)$, in the stochastic game is the expected reward to the patroller in a Nash equilibrium $(\boldsymbol{\Pi}^*, \boldsymbol{\Xi}^*)$, given that the system starts in the state $s$, namely
\begin{equation*}
    \boldsymbol{V}_{pat}(s) = \mathbb{E}_{\boldsymbol{\Pi}^*, \boldsymbol{\Xi}^*} \left[ \sum_{t = 0}^\infty \gamma^t \tilde{R}_{pat}(b_t, \boldsymbol{a}_t, s_t) \; \middle| \; s_0 = s \right]. 
\end{equation*} 
The value of a state $s$ for the smugglers, $\boldsymbol{V}_{smug}(s)$, is defined similarly.
\end{definition}
The value of each state is unique and can be seen to solve the system of equations,
\begin{equation}
    \boldsymbol{V}_{pat}(s) = \max_{\boldsymbol{\pi} \in \Delta([n])} \min_{\boldsymbol{a} \in [0,1]^n} \left[ \sum_{b = 1}^n \pi_b \left\{ \tilde{R}_{pat}(b,\boldsymbol{a},s) + \gamma \boldsymbol{V}_{pat}(b) \right\} \right] \label{lp}
\end{equation}
where $\pi_b$ is the probability the patroller takes action $b$. This follows from \cite{Shapley1953} and \cite{Maitra1970}. By (\ref{transition}), the transitions of system state are determined entirely by the patroller's choice of action. This is why in (\ref{lp}) we can deterministically know the system state resulting from any patroller action.

\cite{Shapley1953} proved that given any initial starting values $\{\boldsymbol{V}_{pat}^0(s) \; | \; s \in \mathcal{S}\}$ the sequence $\{\boldsymbol{V}_{pat}^k(s) \; | \; s \in \mathcal{S} \}_{k=1}^\infty$, determined by the recursion
\begin{equation*}
\boldsymbol{V}_{pat}^k(s) = \max_{\boldsymbol{\pi} \in \Delta([n])} \min_{\boldsymbol{a} \in [0,1]^n} \left[ \sum_{b = 1}^n \pi_b \left\{ \tilde{R}_{pat}(b,\boldsymbol{a},s) + \gamma \boldsymbol{V}_{pat}^{k-1}(b) \right\} \right] 
\end{equation*}
converges to $\{\boldsymbol{V}_{pat}(s) \; | \; s \in \mathcal{S} \}$ as $k \to \infty$. When state and action spaces are finite, state values may be obtained by using linear programming to solve the maximisation problem within Shapley's iteration. However, since we assume that the action space of the smugglers is infinite, this approach is not available to us. Therefore, we look elsewhere to solve (\ref{lp}).

We begin by establishing properties about the smugglers' best response against any patroller strategy. If the smugglers take a best response against patroller strategy $\boldsymbol{\pi} = (\pi_1, \dots, \pi_n)$ when the system state is $s$, the patroller receives a payoff which we shall denote as $G(\boldsymbol{\pi}, s, \boldsymbol{V}_{pat})$
\begin{equation*}
G(\boldsymbol{\pi}, s, \boldsymbol{V}_{pat}) = \min_{ \boldsymbol{a} \in [0,1]^n } \left\{ \sum_{b = 1}^n \pi_b [ \tilde{R}_{pat}(b,\boldsymbol{a},s) + \gamma \boldsymbol{V}_{pat}(b) ] \right\} .
\end{equation*}
If $\boldsymbol{V}_{pat}$ is the value function for the patroller, it will solve the following system of equations 
by (\ref{lp}).
\begin{equation*}
    \boldsymbol{V}_{pat}(s) = \max_{\boldsymbol{\pi} \in \Delta([n])} G(\boldsymbol{\pi}, s, \boldsymbol{V}_{pat}) \text{ for all } s \in [n].
\end{equation*}
\begin{proposition}
We can rewrite $G(\boldsymbol{\pi}, s, \boldsymbol{V}_{pat})$ as:
\begin{equation}
    G(\boldsymbol{\pi}, s, \boldsymbol{V}_{pat}) = \sum_{b=1}^n \left[ - \max_{a_b \in [0,1]} \left\{ (1 - \pi_b) r_b a_b - \pi_b C(a_b) \right\} + \pi_b (\gamma \boldsymbol{V}_{pat}(b) - m_{s,b}) \right] \label{g3}
    \end{equation} 
\end{proposition} 
The set of best responses for the smugglers against patroller strategy when the system state is $s$ is given by,
\begin{equation*}
\boldsymbol{a}(\boldsymbol{\pi}, s) = \argmin_{ \boldsymbol{a} \in [0,1]^n } \left\{ \sum_{b = 1}^n \pi_b [ \tilde{R}_{pat}(b,\boldsymbol{a},s) + \gamma \boldsymbol{V}_{pat}(b) ] \right\} .
\end{equation*}
which we can simplify as a consequence of Proposition 3.
\begin{corollary}
The set of best responses $\boldsymbol{a}(\boldsymbol{\pi}, s)$ can be rewritten as follows:
\begin{equation*}
\boldsymbol{a}(\boldsymbol{\pi}, s) = \argmin_{ \boldsymbol{a} \in [0,1]^n } \left\{ \sum_{b = 1}^n \pi_b [ \tilde{R}_{pat}(b,\boldsymbol{a},s) + \gamma \boldsymbol{V}_{pat}(b) ] \right\} = \left( a_1(\pi_1, s), \dots, a_n(\pi_n, s) \right).
\end{equation*}
where
\begin{equation*}
    a_i (\pi_i, s) = \argmax_{a \in [0,1]} \left\{ (1 - \pi_i)r_i a - \pi_i C(a) \right\}
\end{equation*}
\end{corollary}
\begin{proof}[Proof of Corollary 1]
Follows from taking the argument of the minima in Proposition 3.
\end{proof}
From Corollary 1 we see the smugglers' best response to the patroller is a myopic one, and does not depend on the value function of either player, the discount rate $\gamma$ or the system state $s$.

The function $G$ is additively separable with respect to $\boldsymbol{\pi}$, and so we can write
\begin{equation*}
    G(\boldsymbol{\pi},s,\boldsymbol{V}_{pat}) = \sum_{b=1}^n g_b(\pi_b, s, \boldsymbol{V}_{pat})
\end{equation*}
where
\begin{equation*}
    g_b(\pi_b, s, \boldsymbol{V}_{pat}) = - \max_{a \in [0,1]} \left\{ (1 - \pi_b) r_b a - \pi_b C(a) \right\} + \pi_b (\gamma \boldsymbol{V}_{pat}(b) - m_{s,b})
\end{equation*}
We now develop properties of the functions $g_b$, $b \in [n]$. These will be deployed to develop effecient approaches to the maximisation of $G$.
\begin{lemma} 
For every $b \in [n]$, the function $g_b(\cdot, s, \boldsymbol{V}_{pat}) : \mathbb{R} \to \mathbb{R}$ is concave and Lipschitz continuous on the domain $[0,1]$ with Lipschitz constant $r_b + C(1) - (\gamma \boldsymbol{V}_{pat}(b) - m_{s,b})$ for a fixed system state $s$ and value function $\boldsymbol{V}_{pat}$.
\end{lemma}
\begin{proof}[Proof of Lemma 2]
We utilise \cite{Danskin1967} to establish the convexity of $\max_{a \in [0,1]} \left\{ (1 - \pi_b) r_b a - \pi_b C(a) \right\}$. The concavity of $g_b$ in $\pi_b$, for a fixed $s$ and $\boldsymbol{V}_{pat}$, is then immediate. Lipschitz continuity then follows since $[0,1]$ is compact. See the Appendix for proof of the Lipschitz constant.
\end{proof}
There is an existing literature to solve maximisation problems with an additively separable, concave objective function. Such problems are known as nonlinear knapsack or resource allocation problems. To approximate the continuous problem (\ref{unscaled}) we develop a scaled discrete problem (\ref{scaled}). The scaling factor is denoted $K \in \mathbb{Z}$.

\noindent\begin{minipage}{.5\linewidth}
\begin{align}
\max & \sum_{b=1}^n g_b(\pi_b, s, \boldsymbol{V}_{pat}) \nonumber \\ 
s.t. & \sum_{b=1}^n \pi_b = 1 \label{unscaled} \\
& \pi_b \in [0,1] \nonumber
\end{align}
\end{minipage}%
\begin{minipage}{.5\linewidth}
\begin{align}
\max & \sum_{b=1}^n g_b(\pi_b, s, \boldsymbol{V}_{pat}) \nonumber \\
s.t. & \sum_{b=1}^n \pi_b = 1 \label{scaled} \\
& \pi_b \in \left\{ \frac{0}{K}, \frac{1}{K}, \dots, \frac{K}{K} \right\} \nonumber
\end{align}
\end{minipage}

We denote by $\boldsymbol{\pi}^*$ the optimal solution for the continuous problem (\ref{unscaled}) and by $\tilde{\boldsymbol{\pi}}_K$ the approximate solution obtained from the discrete scaled problem (\ref{scaled}). Scaling by $K = n/\delta$ for some small $\delta > 0$ gives us the bound that $\| \boldsymbol{\pi}^* - \tilde{\boldsymbol{\pi}}_{n/\delta} \|_{\infty} \leq \delta$ by the proximity result of \cite{Hochbaum1994}. Therefore, since for all $b$ the function $g_b$ is Lipschitz continuous we have that,
\begin{align*}
    | G(\boldsymbol{\pi}^*, s, \boldsymbol{V}_{pat}) - G(\tilde{\boldsymbol{\pi}}_{n/\delta}, s, \boldsymbol{V}_{pat}) | \leq & \sum_{b=1}^n |g_b(\pi_b^*, s, \boldsymbol{V}_{pat}) - g_b(\tilde{\pi}_{n/\delta, b}, s, \boldsymbol{V}_{pat})| \\
    \leq & \sum_{b=1}^n [r_b + C(1) - (\gamma \boldsymbol{V}_{pat}(b) - m_{s,b})] |\pi^*_b - \tilde{\pi}_{n/\delta, b}| \\
    \leq & \delta \sum_{b=1}^n [r_b + C(1) - (\gamma \boldsymbol{V}_{pat}(b) - m_{s,b})] .
\end{align*}
We conclude that $|G(\boldsymbol{\pi}^*, s, \boldsymbol{V}_{pat}) - G(\tilde{\boldsymbol{\pi}}_{n/\delta}, s, \boldsymbol{V}_{pat})| = \mathcal{O}(n\delta)$. The discrete resource allocation problem (\ref{scaled}) can be solved greedily, as shown by \cite{Fox1966}. This yields in Algorithm 1 for its solution.

\begin{algorithm}[]
\SetAlgoLined
\KwInitialise{$\tilde{\boldsymbol{\pi}}_K = (0, \dots, 0), k = 0$}
\While{$k < 1$}{
Let,
\begin{equation*}
    j \in \argmax_{b \in [n]} \left\{ g_b \left(\tilde{\pi}_b + \frac{1}{K}, s, \boldsymbol{V}_{pat}\right) - g_b(\tilde{\pi}_b, s, \boldsymbol{V}_{pat}) \right\}
\end{equation*}
with ties decided by taking the lowest index.\\
$\tilde{\pi}_{K,b} := \tilde{\pi}_{K,b} + \frac{1}{K}$ and $k := k + \frac{1}{K}$ 
}
\KwOutput{$\boldsymbol{\tilde{\pi}}_K$}
\caption{Greedy Procedure by \cite{Fox1966}}
\end{algorithm}

While the complexity of Algorithm 1 is $\mathcal{O}(Kn) = \mathcal{O}(n^2/\delta)$ there exist more computationally efficient algorithms in the literature, such as \cite{Kaplan2019}. This has a computational complexity of $\mathcal{O}(n\log{K}) = \mathcal{O}(n \log (n/\delta))$. In our examples, we consider both the algorithm by \cite{Fox1966} and by \cite{Kaplan2019}. We have found that which is the quicker algorithm in practise can depend on the parameters of the problem.

We now leverage our ability to find a $\delta$-optimal solution to the problem (\ref{unscaled}) in order to find the values of the states in the game via the iterative method of \cite{Shapley1953}. This yields Algorithm 2.

\begin{algorithm}[]
\SetAlgoLined
\KwInput{$\epsilon > 0$ and $\delta > 0$ }
\KwInitialise{$\boldsymbol{V}_{pat}^0(s) = (0,\dots,0)$ and $k = 1$} 
\While{$\max_{s \in \mathcal{S}} \left\{ \left| \boldsymbol{V}_{pat}^{k-1}(s) - \boldsymbol{V}_{pat}^k(s) \right| \right\} > \epsilon$}{
    \For{$s = 1, \dots, n$}{
        Find,
        \begin{align*}
        \boldsymbol{V}_{pat}^k(s) &:= \max_{\boldsymbol{\pi} \in \Delta([n])} \min_{\boldsymbol{a} \in [0,1]^n} \left[ \sum_{b = 1}^n \pi_b \left\{ \tilde{R}_{pat}(b,\boldsymbol{a},s) + \gamma \boldsymbol{V}^{k-1}_{pat}(b) \right\} \right] \\
        &= \max_{\boldsymbol{\pi} \in \Delta([n])} G(\boldsymbol{\pi}, s, \boldsymbol{V}_{pat}^{k-1})
        \end{align*}
        using Algorithm 1 with $K = n/\delta$.
    }
        $k := k + 1$ 
}
\KwOutput{$\boldsymbol{V}_{pat}^k$}
\caption{Calculation of state values}
\end{algorithm}

Once the state values have been calculated, the patroller's strategy $\boldsymbol{\Pi}^*$ can be identified as the value of $\tilde{\boldsymbol{\pi}}$ found in Step 3 of Algorithm 2. However, finding the smugglers' strategy $\boldsymbol{\Xi}^*$ which forms a Nash equilibrium with $\boldsymbol{\Pi}^*$ is a complex task without the addition of further assumptions on the parameters. In the next section we explore the characteristics of Nash equilibria under additional assumptions.

\section{Behaviour of the Smugglers' Best Response}

In this section, we focus on two different assumptions about the cost function $C$, which quantifies the losses of the smuggler when caught by the patroller. When the cost function is concave, we show that in a Nash equilibrium the smugglers only take actions in $\{0,1\}$. This yields a more computationally efficient algorithm than Algorithm 2 in such cases, which is also guaranteed to find the optimal solution $\boldsymbol{\pi}^*$. When $C$ is a strictly convex function, we show that the smugglers' strategy in equilibria takes actions deterministically.

\subsection{Concave Cost Functions}

We first examine the case in which the cost function $C$ is a linear function, and then proceed to the case in which it is strictly concave. Under linearity, we prove that at least one of the actions zero or one lies within the set of best responses for each smuggler. Recall that we can calculate the set of best responses to patroller strategy $\boldsymbol{\pi}$ for the smuggler at a location $b$ by
\begin{equation*}
   a_b(\pi_b, s) = \argmax_{a \in [0,1]} \left\{ (1 - \pi_b)r_b a - \pi_b C(a) \right\} .
\end{equation*}
\begin{proposition}
If $C$ is concave, either $0 \in a_b(\pi_b,s)$ or $1 \in a_b(\pi_b,s)$. Furthermore if there exists $a \in (0,1)$ such that $a \in a_b(\pi_b,s)$, then $C$ must be linear. 
\end{proposition}
\begin{proof}[Proof of Proposition 4]
The function $(1-\pi_b) r_b a - \pi_b C(a)$ is convex in $a$, since $C$ is concave. A maxima of a convex function on a convex set can always be found at an extreme points of that set, establishing the first result. If a maxima exists in the interior of the set, then the function must be constant on the set. In the case that $(1-\pi_b)r_b a - \pi_b C(a)$ is constant, $C$ must be linear.
\end{proof}
Proposition 4 allows us to simplify the game by reducing the action space of the smugglers.
\begin{corollary}
If the cost function $C$ is concave and $(\boldsymbol{\Pi}^*, \boldsymbol{\Xi}^*)$ is a Nash equilibrium in the border patrol game, then there exists a strategy $\tilde{\boldsymbol{\Xi}} \in (\Delta(\{0,1\}))^n$ such that $(\boldsymbol{\Pi}^*, \tilde{\boldsymbol{\Xi}})$ is a Nash equilibrium.
\end{corollary}
\begin{proof}[Proof of Corollary 2]
Suppose that the strategy $\boldsymbol{\Xi}^*$ takes an action $\boldsymbol{a}$ where $a_b \in (0,1)$ for some $b \in [n]$ with positive probability. Since $\boldsymbol{\Xi}^*$ must be a best response to $\boldsymbol{\Pi}^*$, Proposition 4 implies that $C$ must be linear.
The patroller's best response to $\boldsymbol{\Xi}^*$ when the system state is $s$ gives a payoff of,
\begin{equation*}
    \max_{b \in [n]} \left\{\mathbb{E} \left[ \tilde{R}_{pat}(b, \boldsymbol{a}, s) + \gamma V_{pat}(b) \right] \right\}
\end{equation*}
where the expectation is taken over $\boldsymbol{a} \sim \boldsymbol{\xi}_s$. However, when $C$ is linear we have that
\begin{align*}
    \mathbb{E} \left[ \tilde{R}_{pat}(b, \boldsymbol{a}, s) + \gamma V_{pat}(b) \right] &= \mathbb{E} \left[ C(a_b) - \sum_{i \in [n] \setminus \{b\}} r_i a_i - m_{s,b} + \gamma V_{pat}(b) \right] \\
    &= c \mathbb{E}[a_b] - \sum_{i \in [n] \setminus \{b\}} r_i \mathbb{E}[a_i] - m_{s,b} + \gamma V_{pat}(b) 
\end{align*}
for some $c > 0$. Therefore, as if the smugglers instead take a strategy over the actions zero and one such that the expected quantity remains constant, then both players receive the same expected payoff. Hence, neither the patroller nor smuggler has incentive to deviate and so is a Nash equilbrium. This concludes the proof. 
\end{proof}
From Corollary 2 we infer that the smuggler action space can be reduced to $\mathcal{A} = \{0,1\}^n$ without loss of generality. Having a finite action space for the smugglers means that the stochastic game is now finite and so Nash equilibria can be found using a linear programming formulation for single controller stochastic games. This is as in \cite{Raghavan2003}. Alternatively, we can use linear programming to maximise (\ref{lp}) in the iterative algorithm by \cite{Shapley1953}. This is as in \cite{Filar1997}. Corollary 2 also means that we can replace a strictly concave cost function with a linear cost function, provided that it takes the same values at the endpoints zero and one.
\begin{corollary}
    If the cost function $C$ is a concave function then the Nash equilibria are equivalent those in a game with identical parameters, but a cost function $\tilde{C}$ defined by $\tilde{C}(a) = C(1)a$.
\end{corollary}
\begin{proof}[Proof of Corollary 3]
    By Corollary 2, we have that the smuggler action space is $\{0,1\}^n$. Therefore, the cost function $C$ is evaluated only at the points $a \in \{0,1\}$. Since $\tilde{C}(0) = C(0)$ and $\tilde{C}(1) = C(1)$, any Nash equilibria in the game with the cost function $C$ must also be Nash equilibria in the game with the cost function $\tilde{C}$.  
\end{proof}
We now look to simplify the function $G$ given by (\ref{g3}).
\begin{lemma}
Assuming that the cost function $C$ is concave, we can write $G$ as
\begin{equation}
    G(\boldsymbol{\pi}, s, \boldsymbol{V}_{pat}) = \sum_{b=1}^n \left\{ [\pi_b(C(1) + r_b) - r_b] \mathbbm{1} \left( \pi_b \leq \frac{r_b}{C(1) + r_b} \right) + \pi_b (\gamma \boldsymbol{V}_{pat}(b) - m_{s,b}) \right\} \label{g2} .
\end{equation}
\end{lemma}
\begin{proof}[Proof of Lemma 3]
    Proposition 4 implies that $0 \in a_b(\pi_b, s)$ or $1 \in a_b(\pi_b, s)$. If we evaluate the smuggler's payoff at the two we get,
    \begin{align*}
        a = 0 &\implies (1-\pi_b) r_b a - \pi_b C(a) = 0 \\
        a = 1 &\implies (1-\pi_b) r_b a - \pi_b C(a) = (1 - \pi_b) r_b - C(1).
    \end{align*}
    This means that,
    \begin{equation*}
       \mathbbm{1} \left( \pi_b^s \leq \frac{r_b}{C(1) + r_b} \right) \in a_b(\pi_b, s)
    \end{equation*}
    and so,
    \begin{equation*}
        \max_{a \in [0,1]} \{ (1 - \pi_b) r_b a - \pi_b C(a) \} = [(1 - \pi_b)r_b - \pi_b C(1)] \mathbbm{1} \left( \pi_b^s \leq \frac{r_b}{C(1) + r_b} \right) 
    \end{equation*}
    Substituting this into the expression for $G$ in Equation (\ref{g3}) gives the result in the statement of the lemma.
\end{proof}
A consequence of Lemma 3 is that we can now provide a computationally efficient method in Algorithm 3 to find the optimal value of $G$ when the cost function $C$ is concave.
\begin{algorithm}[]
\SetAlgoLined
\KwInitialise{$\hat{\boldsymbol{\pi}} = (0, \dots, 0)$}
\While{$\sum_{b=1}^n \pi_b < 1$}{
Define for all $b$,
\begin{equation*}
    x_b = \begin{dcases}
    \frac{r_b}{C(1) + r_b} & \text{ if } \hat{\pi}_b = 0 \\
    1 - \frac{r_b}{C(1) + r_b} & \text{ otherwise.}
    \end{dcases}
\end{equation*} \\
Choose arbitrarily,
\begin{equation*}
j \in \argmax_{b \in [n]} \left\{ \frac{g_b(\hat{\pi}_b + x_b) - g_b(\hat{\pi}_b)}{x_b} \right\}
\end{equation*} \\
with ties decided by taking the lowest index. \\
\eIf{$\sum_{b=1}^n \hat{\pi}_b + x_j \leq 1$}{
    Let $\hat{\pi}^j := \hat{\pi}^j + x_j$ . \\
 }{
    Let $\hat{\pi}^j := \hat{\pi}^j + \left( 1 - \sum_{b=1}^n \hat{\pi}_b \right) . $
}
}
\KwOutput{$\hat{\boldsymbol{\pi}}$}
\caption{Concave Cost Greedy maximization of $G$}
\end{algorithm}

\begin{theorem}
If $C$ is linear then $\hat{\boldsymbol{\pi}}$, the output of Algorithm 3 maximises $G$.
\end{theorem}

The complexity of Algorithm 3 is only $\mathcal{O}(n)$, since the maximum number of iterations needed to complete is $n+1$ and each iteration has complexity $\mathcal{O}(1)$. We can see that it takes at most $n+1$ iterations, since once the probability of an action is increased twice, the algorithm must terminate.

So far, the discussion has focused only on determining a strategy $\boldsymbol{\Pi}^*$ for the patroller. We now consider how to find a strategy for the smugglers $\boldsymbol{\Xi}^*$ such that ($\boldsymbol{\Pi}^*, \boldsymbol{\Xi}^*)$ is a Nash equilibrium in our model. Once we have found the value function $\boldsymbol{V}_{smug} = - \boldsymbol{V}_{pat}$, finding the smugglers' strategy can be found by taking:
\begin{equation*}
    \boldsymbol{\xi}^s \in \argmax_{\boldsymbol{\xi}^s \in \Delta(\{0,1\}} \min_{b \in [n]} \mathbb{E} \left[ \tilde{R}_{smug}(b, \boldsymbol{a}, s) + \gamma V_{smug}(b)\right]
\end{equation*} 
for each $s \in [n]$. A linear program can efficiently solve this as in \cite{Filar1997}.

\subsection{Strictly Convex Cost Function}

We now proceed to the case in which the cost function $C$ is strictly convex with respect to the action taken by the smugglers. Algorithm 1 can give us an approximation for $\boldsymbol{\Pi}^*$, but as in the previous subsection, we still need to consider how we will calculate the smugglers' strategy $\boldsymbol{\Xi}^*$. The following lemma shows that under an assumption of strict convexity there can only be one choice, and additionally it is simple to find.
\begin{lemma}
If $C$ is strictly convex, then for any given patroller strategy $\boldsymbol{\pi}$, the smugglers have a single best response.
\end{lemma}
\begin{proof}[Proof of Lemma 4]
Recall that when the system state is $s$, the set of best responses for the smuggler at the location $b$ to the patroller's strategy $\boldsymbol{\pi}$ is given by,
\begin{equation}
a_b(\pi_b, s) = \max_{a \in [0,1]} \left\{ (1 - \pi_b) r_b a - C(a) \right\} . \label{con}
\end{equation}
If $C$ is strictly convex, then the function to be maximised in (\ref{con}) is strictly concave in $a_b$. A strictly concave function can only have a single maximum in the interval $[0,1]$, and therefore there can only be a single unique best response for the smuggler at $b$ for any given patroller strategy. Applying this reasoning to each system state and every location, we can see that there must be a single strategy for the smugglers which is uniquely the best response to $\boldsymbol{\pi}$.
\end{proof}
From Lemma 4, we can quickly compute a best smuggler response $\boldsymbol{\Xi}$ to the patroller's strategy $\boldsymbol{\Pi}^*$. Since there exists a best response to $\boldsymbol{\Pi}^*$, and since there must exist at least one Nash equilibrium, then $(\boldsymbol{\Pi}^*,\boldsymbol{\Xi})$ must indeed be a Nash equilibrium.
\section{Examples}
In this section, we introduce three different examples and discuss how the analysis from previous sections helps to find Nash equilibria and how to understand them. We then go on to justify the use of a stochastic game model in terms of its benefits for the border patrol problem compared to the use of alternative models.
\subsection{Example 1: Linear Border With Linear Cost Function}
We begin by considering an example with a linear cost function. We compare the time taken to find Nash equilibria using the methods discussed in this paper with existing methods in the literature. We can apply the latter, since by Corollary 2 we know that there exists a Nash equilibrium in which the smugglers' actions are supported by $\{0,1\}^n$.

We consider a cost function of $C(a) = 4a$. The reward to each smuggler for success is just the amount of items they send, so that $r_i = 1$ for every location $i$. We define the movement cost for the patroller to be $m_{i,j} = |i-j|^2$. The number of locations in the border shall be varied to display how the methods scale with the size of the problem. Finally, we consider a fixed discount factor of $\gamma = 0.9$ for each player.

In Table 1, we present the time it takes for five different algorithms to find a Nash equilibrium in the model. The first method is to solve a single linear program using the formulation of \cite{Raghavan2003} for single-controller stochastic games. The other methods use the iterative method of \cite{Shapley1953} in Algorithm 2 with different methods to find the solution to the maximisation problem in Step 3. The first of these deploys a linear program using a formulation by \cite{Filar1997}. Subsequent approaches solve it as a resource allocation problem using the algorithms of \cite{Fox1966} and \cite{Kaplan2019}. The final method reported solves using our method assuming a linear cost function in Algorithm 3. We set the tolerance $\epsilon$ in Algorithm 1 to $10^{-3}$, and the scaling of the resource allocation to $\delta = 0.2$. Note that since $K = n/\delta = 5n$, it is always divisible by $r_b + C(1) = 5$. Therefore, by Theorem 1 the resource allocation problem finds the optimal solution.

Table 1 shows that Algorithm 3 dramatically speeds up the calculation of a Nash equilibrium in our game having a 400\%, 1800\%, 13000\% and 43000\% improvement in each respective example over the next best method. We take the case with $n = 6$ locations and show the patroller's strategy for a Nash equilibrium in Figure 1.

We see in Figure 1(a) an illustration of the result of Lemma 3 and Theorem 1, with the patroller choosing actions with probability in multiples of $0.2 = r_b/(r_b+C(1))$. Similarly in Figure 1(b) we see that the smuggler's best response is to send an item with probability one, with probability zero or an intermediate value if the location is protected with respectively a probability less than, greater than or exactly 0.2. Note that as a consequence of Corollary 3, the results in Example 1 countinue hold if we had a concave cost function $C$ taking the values $C(0) = 0$ and $C(1) = 4$.

\subsection{Example 2: Linear Border With Strictly Convex Cost Function}
We now give an example with a strictly convex cost function and show how this yields a different solution from the previous example. The parameters of the model are identical to those in Example 1 ($r_i = 1$ for all $i$, $m_{i,j} = |i-j|^2$, $\gamma = 0.9$), except now we take $C(a) = 4a^2$. Note that we still have $C(0) = 0$ and $C(1) = 4$ as before, and so in our results demonstrate that the simplifications afforded in Example 1 for the concave case no longer apply.

The strategies obtained in this subsection are not necessarily Nash equilibria, since by using the discretization of the resource allocation problem in (\ref{scaled}) we derive only a $\delta$-optimal solution. We assess the closeness to equilibrium by examining the worst case expected reward to the patroller under their strategy $\boldsymbol{\Pi}$. We calculate the worst case expected reward (WCER) by finding a strategy for the smugglers $\boldsymbol{\Xi}^*$ that is the best response to the patroller's strategy $\boldsymbol{\Pi}$. Assuming a uniform distribution over the inital state of the system, $\mathbb{P}(s_0 = s) = 1/n$, we can calculate the WCER for the patroller as follows.
\begin{align*}
    WCER(\boldsymbol{\Pi}) &= \min_{\boldsymbol{\Xi} \in (\Delta([0,1]^n))^n} \left\{ \frac{1}{n}\sum_{s = 1}^n  \mathbb{E}_{\boldsymbol{\Pi}, \boldsymbol{\Xi}} \left[ \sum_{t = 0}^\infty \gamma^t \tilde{R}_{pat}(b_t, \boldsymbol{a}_t, s_t) \; \middle| \; s_0 = s \right]  \right\} \\
    &= \frac{1}{n}\sum_{s = 1}^n \mathbb{E}_{\boldsymbol{\Pi}, \boldsymbol{\Xi}^*} \left[ \sum_{t = 0}^\infty \gamma^t \tilde{R}_{pat}(b_t, \boldsymbol{a}_t, s_t) \; \middle| \; s_0 = s \right].     
\end{align*}
Note that,
\begin{equation*}
    WCER(\boldsymbol{\Pi}^*) = \frac{1}{n} \sum_{s = 1}^n \boldsymbol{V}_{pat}(s) \geq WCER(\boldsymbol{\Pi}).
\end{equation*}
The two methods that we implement to solve this example are the resource allocation algorithms of \cite{Fox1966} and \cite{Kaplan2019} within the iterative algorithm of \cite{Shapley1953}. The resource allocation problem now becomes more challenging to solve, compared to the linear case, since the smugglers' best response to the patroller is more complex. Therefore, there is no scaling of the continuous problem (\ref{unscaled}) that will give us the optimal solution. Now, the smaller the choice of $\delta$, the better the strategy $\boldsymbol{\Pi}$ computed. In Table 2, we give the time taken and worst case reward for the two algorithms under different choices of scaling. As in Table 1, a tolerance of $\epsilon = 10^{-3}$ was used for Algorithm 2.

In Table 2, we can see that as $\delta$ gets smaller the worst case expected reward improves for the patroller. Tables 3 and 4 give the time taken to compute the strategies shown in Table 2.

In the cases with few locations and low fidelity of scaling, the algorithm by \cite{Fox1966} is quicker than that of \cite{Kaplan2019} but as the problem size grows this is no longer the case. Note how in the six location example changing the scaling factor has a much bigger effect on the worst case expected reward of the strategy than in the fifteen location problem. In Figure 2, we show the strategy calculated for $\delta = 0.04$ and six locations.

We can see that the strategy given in Figure 2(a) is quite different to that in Figure 1. No longer are the probabilities multiples of $r_b/(r_b + C(1))$. The patroller is now less likely to move away from one of the two edges of the border, a key impact that changing the cost function has had on their decision making. In Figure 2(b), there is also a large difference in the strategy displayed compared to Figure 1(b), having a single action in $[0,1]$ taken with probability one by each smuggler. These differences elucidate the importance we ascribe to the modelling of costs in our analysis.
\subsection{Example 3: Perimeter Border With Linear Cost Function}
We now turn our attention to an alternative border structure that is important operationally, namely a circular perimeter of an area. We now define the movement cost as the minimum of the length of the two paths the patroller could take between locations. This yields $m_{i,j} = \min \{ |i-j|, n - |i-j| \}$ for $i, j \in [n]$. We also consider a setup in which rewards for the smugglers are location dependent. In reality, there could be various reasons for this including the difficulty in getting through the border and the value of the items on the other side. Here, we set the rewards $\boldsymbol{r}$ equal to $(3,2,1,1,2,3)$. The remaining parameters of the model those of the first example $( C(a) = 4a,$ $\gamma = 0.9)$. Equilibria in this example are computed as for Example 1.

Figure 3(a) shows how the locations protected most heavily are those with higher smuggler reward, which in this example are locations one and six. Note that since the locations form a circle, the patroller can travel from location one to location six at a cost of one unit. This is one reason why Figure 3 (a) looks different from the patroller strategies in previous examples. We continue with the pattern of Figure 1(a), namely that the patroller protects location $b$ with probability $r_b/(r_b+C(1))$. The other values in Figure 3(a) arise as a result of the probabilities needing to sum to one.
\subsection{Value Of Modelling as a Stochastic Game}
In this subsection, we will evaluate the benefits of using our model over alternative modelling approaches. We consider how the patroller's worst case expected reward would be affected if the game was considered to be normal-form, with no consideration to the state of the system in the next time step. We achieve this by finding the patroller's strategy in a Nash equilibrium when we set the discount rate to $\gamma$ to zero.  

The first case we consider is one in which all movement costs are set to zero. This means that the state of the system no longer has an effect on the rewards to either player. We make this choice since the assumptions in the work of \cite{Pita2009} or \cite{Alpern2011} are similar. However, the patroller will be accumulating costs to travel without knowing of their existence. To make a fairer comparison, we consider a second case where movement costs are considered by the patroller but she still acts myopically. This is equivalent to solving a normal-form game for each state in which the patroller could start. 

In Figure 4, we show the strategies obtained under these two sets of assumptions for the model given in Example 3. Figure 4(a) shows why it is important to consider the geography underlying the model, with the patroller making large moves at a high cost. In Figure 4(b) we overcome this, but the patroller can still be seen making suboptimal moves because she is not accounting for the value of the state to which she moves. In Table 5, we take the six location version of the three examples introduced previously in the section, find the strategies as detailed in the previous paragraph and calculate the worst case expected reward for the patroller in each case. 

Having no consideration for the states in the game leads to a large decrease in the reward to the patroller. In Example 1, for example, she incurs over twice the cost than in the full model. Factoring in movement costs but still disregarding future rewards improves the outcome to the patroller, but there is still a very significant benefit to solving with the full stochastic game model.The computational challenge of developing solutions may have been grounds for the earlier focus on over-simple models. Our analysis removes many of these obstacles.

\section{Conclusion}
In this paper, we have introduced a new model which can be used to consider the interaction between smugglers and a patroller on a border. A number of properties of Nash equilibria in the stochastic game are established, and new algorithms to find these equilibria are developed. We provide examples to show empirically that our methods solve the model quicker than existing methods and additionally that using a stochastic game formulation achieves significant improvement for the patroller.

\section{Acknowledgements}
This paper is based on work completed while Matthew Darlington was part of the EPSRC funded STOR-i centre for doctoral training (EP/S022252/1). For the purpose of open access, the author has applied a Creative Commons Attribution (CC BY) licence [where permitted by UKRI, ‘Open Government Licence’ or ‘Creative Commons Attribution No -derivatives (CC BY-ND) licence’ may be stated instead] to any Author Accepted Manuscript version arising

\printbibliography

\section{Tables}

\begin{table}[H]
    \centering
    \caption{Time taken (secs.) to solve Example 1 with different numbers of locations $n$}
    \label{tab:my_label}
\scalebox{.7}{
\begin{tabular}{ccccc}
\toprule
Algorithm Used                                  & $n = 6$         & $n = 9$         & $n = 12$        & $n = 15$        \\ \midrule
Single Controller Linear Program       & \num{0.250}     & \num{2.297}     & \num{25.563}    & \num{766.578}   \\
Shapley method with linear programming         & \num{2.00}      & \num{47.125}    & \num{338.188}   & \num{2812.219}  \\ 
Shapley with resource allocation (Fox 1966)       & \num{19.016}    & \num{52.353}    & \num{78.516}    & \num{120.281}   \\
Shapley with resource allocation (Kaplan et al. 2019)    & \num{19.859}    & \num{50.000}    & \num{103.094}   & \num{132.609}   \\ 
Shapley with Algorithm 3                        & \num{0.063}    & \num{0.125}     & \num{0.203}     & \num{0.281}     \\ \bottomrule
\end{tabular}
}
\end{table}

\begin{table}[H]
    \centering
    \caption{Worst case expected reward in Example 2}
\scalebox{.7}{
    \begin{tabular}{ccccc}
\toprule
                     & $\delta = 1$    & $\delta = 0.2$  & $\delta = 0.1$  & $\delta = 0.04$ \\ \midrule                               
$n = 6$              & -39.068         & -38.338         & -38.291         & -38.282         \\
$n = 9$              & -67.740         & -67.571         & -67.551         & -67.544         \\
$n = 12$             & -97.681         & -97.239         & -97.230         & -97.227         \\
$n = 15$             & -127.200        & -127.060        & -127.052        & -127.049        \\ \bottomrule
\end{tabular}
}
\end{table}
\begin{table}[H]
    \centering
    \caption{Time taken (s) to solve Example 2 (Fox 1966)}
\scalebox{.7}{
\begin{tabular}{ccccc}
\toprule
                     & $\delta = 1$    & $\delta = 0.2$  & $\delta = 0.1$  & $\delta = 0.04$ \\ \midrule  
$n = 6$              & 6.688           & 16.125          & 29.453          & 82.875          \\ 
$n = 9$              & 17.484          & 44.844          & 84.281          & 192.391         \\
$n = 12$             & 31.359          & 91.844          & 159.641         & 337.500         \\
$n = 15$             & 45.375          & 139.984         & 603.641         & 1407.313        \\  \bottomrule
\end{tabular}
}
\end{table}
\begin{table}[H]
    \centering
    \caption{Time taken (s) to solve Example 2 (Kaplan et al. 2019)}
\scalebox{.7}{
    \begin{tabular}{ccccc}
\toprule
                    &  $\delta = 1$     & $\delta = 0.2$   & $\delta = 0.1$  & $\delta = 0.04$   \\ \midrule      
$n = 6$              & 6.641            & 18.656           & 24.563          & 30.719          \\ 
$n = 9$              & 19.672           & 50.469           & 69.016          & 83.156          \\
$n = 12$             & 31.453           & 118.344          & 128.234         & 147.734         \\
$n = 15$             & 49.844           & 143.313          & 498.938         & 591.016         \\ \bottomrule
\end{tabular}
}
\end{table}

\begin{table}[H]
\centering
\caption{Worst Case Expected Rewards Under A Range of Models}
\scalebox{.7}{
\begin{tabular}{cccc}
\toprule
Model                               & Example 1 & Example 2 & Example 3 \\ \midrule
Normal-form Game (without movement cost) & -68.333    & -73.958    & -64.238    \\
Normal-form Game (with movement cost)    & -34.000    & -38.743    & -61.189    \\
Stochastic Game                          & -33.587    & -38.282    & -60.110    \\ \bottomrule
\end{tabular}
}
\end{table}

\newpage
\section{Figures}

\begin{figure}[H]
\centering
\begin{subfigure}{.5\textwidth}
  \centering
  \includegraphics[width=.95\linewidth]{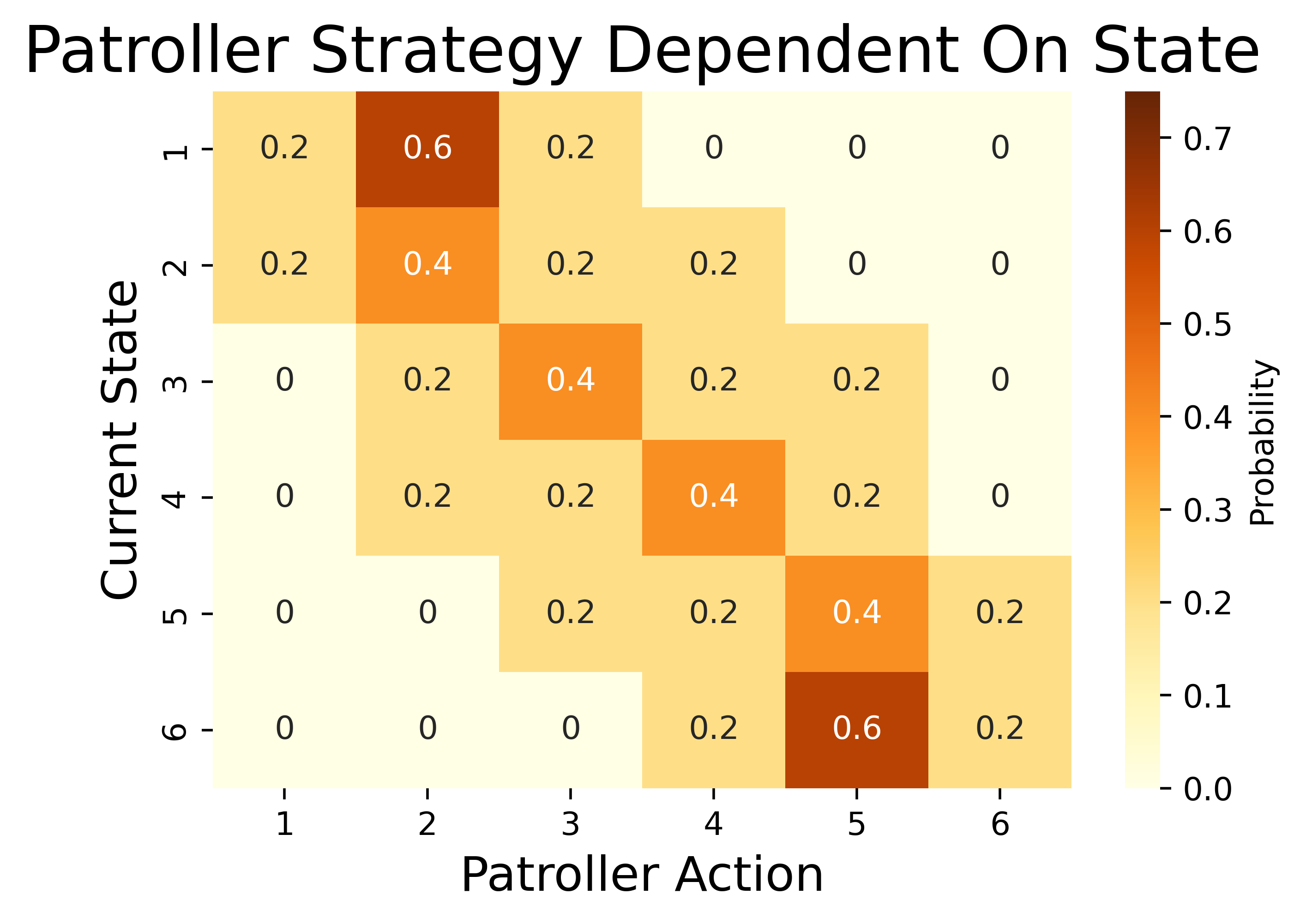}
  \caption{}
\end{subfigure}%
\begin{subfigure}{.5\textwidth}
  \centering
  \includegraphics[width=.95\linewidth]{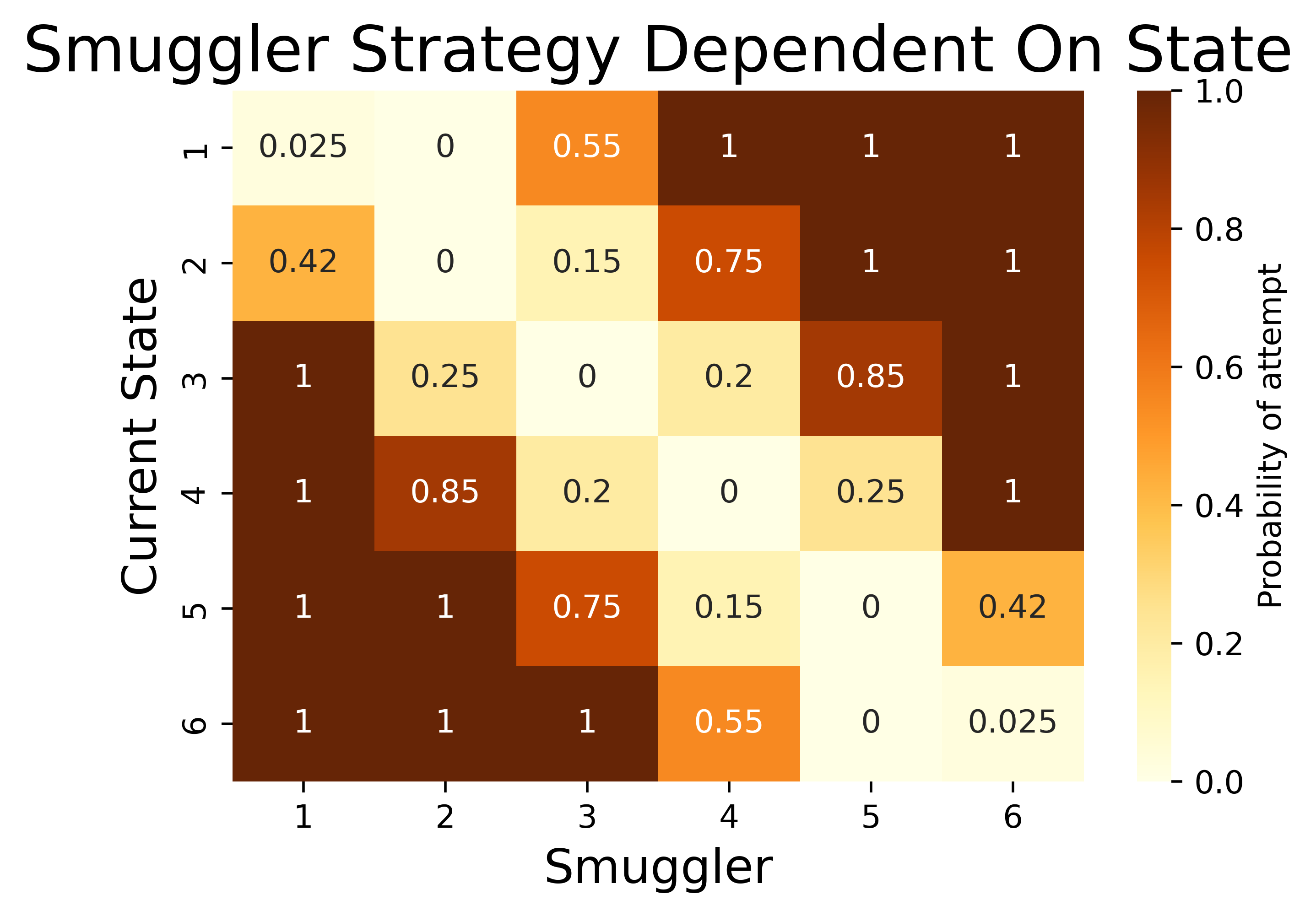}
  \caption{}
\end{subfigure}
\caption{A Nash equilibrium in Example 1. The vertical axis gives the current state $s$ of the system in both figures. In (a) the horizontal axis shows each location the patroller could move to and the colour gives the probability with which they take that action. In (b) the horizontal axis gives each smuggler and the colour gives the probability with which they make an attempt to smuggle an item.}
\end{figure}
\vspace{-8mm}

\begin{figure}[H]
\centering
\begin{subfigure}{.5\textwidth}
  \centering
  \includegraphics[width=.95\linewidth]{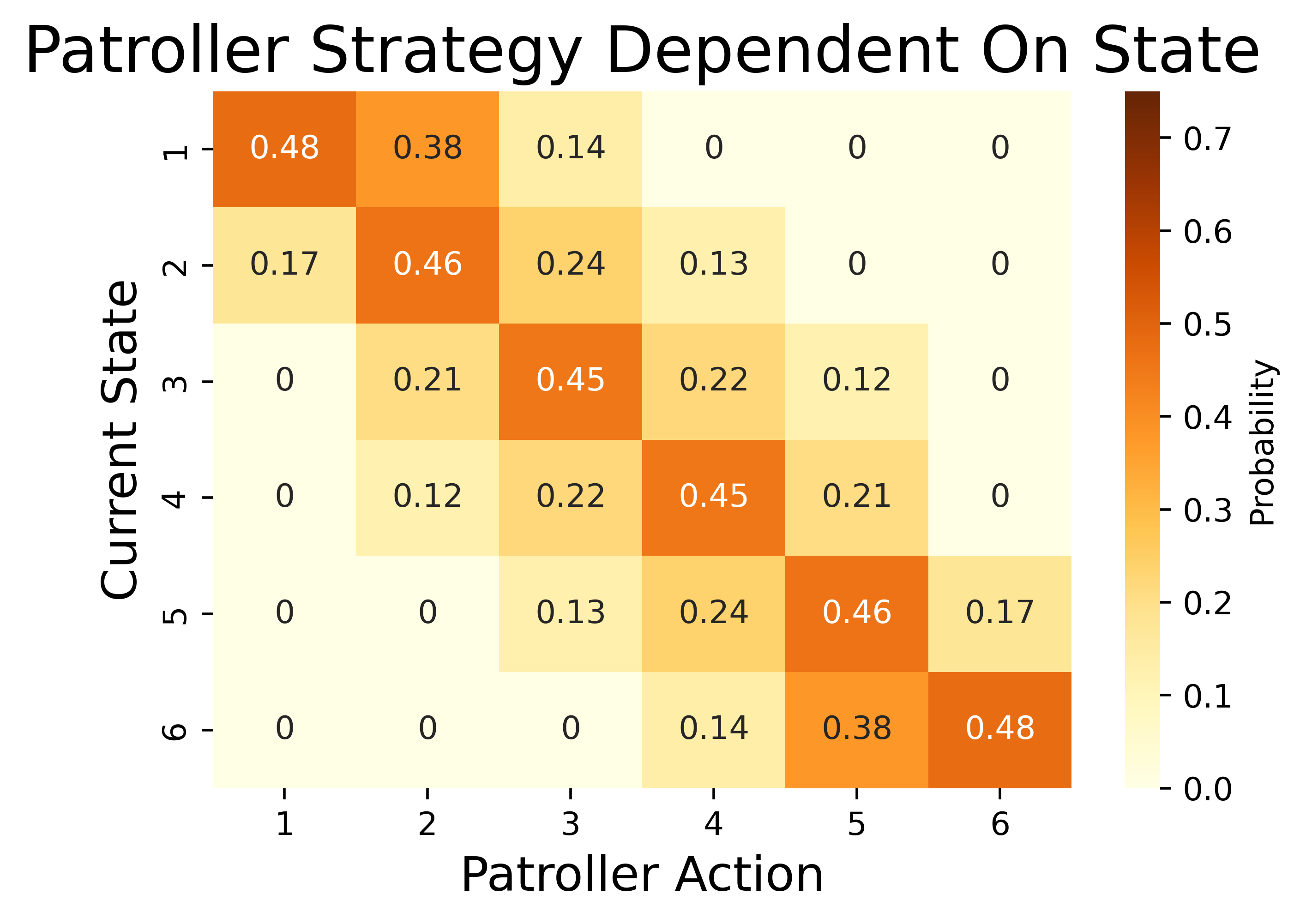}
  \caption{}
\end{subfigure}%
\begin{subfigure}{.5\textwidth}
  \centering
  \includegraphics[width=.95\linewidth]{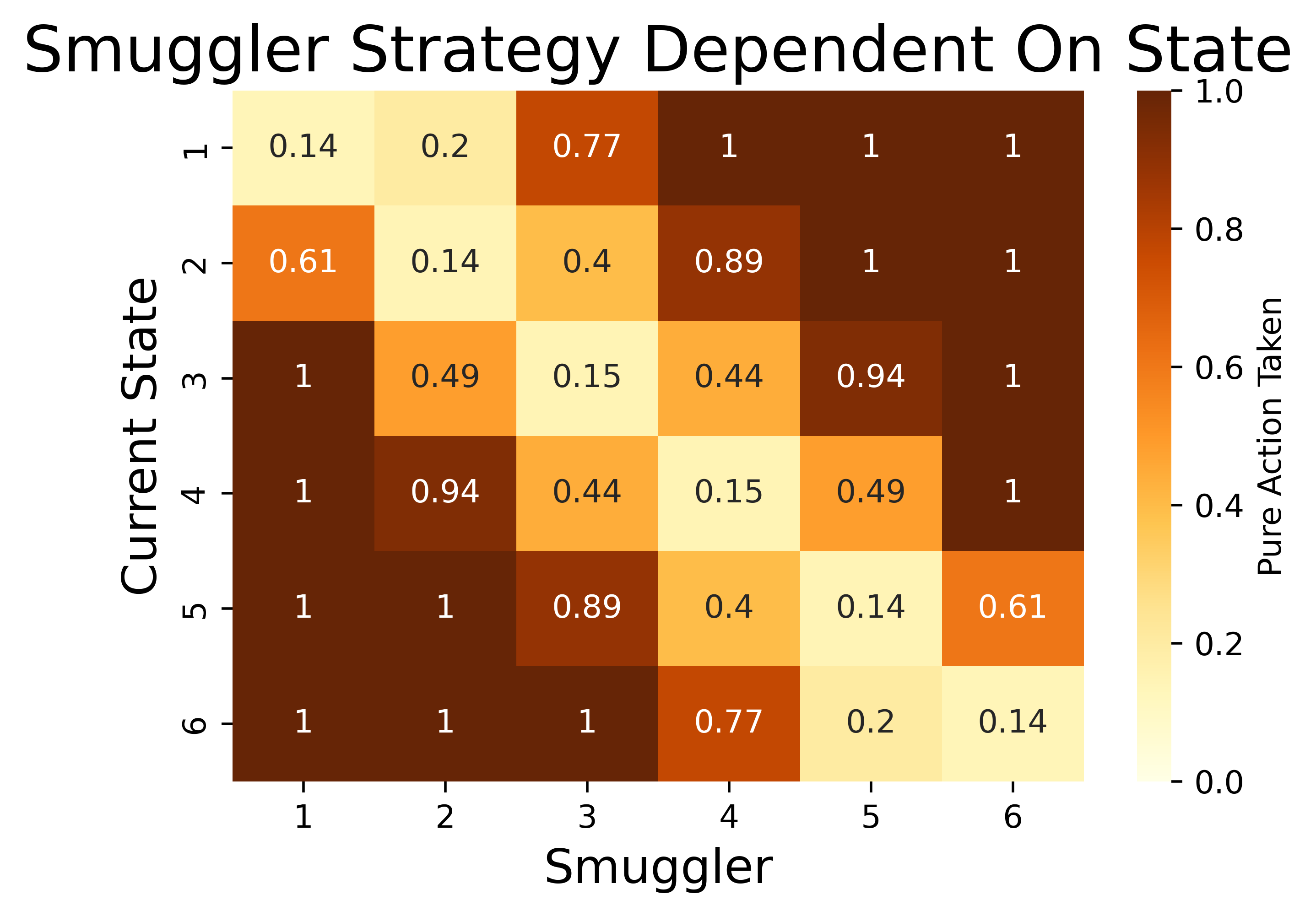}
  \caption{}
\end{subfigure}
\caption*{A Nash equilibrium in Example 2. The vertical axis gives the current state $s$ of the system in both figures. Figure 2(a) has the same interpretation as in Figure 1(a). In (b) the horizontal axis denoted each smuggler and the colour now gives the quantity of items they attempt to smuggle with probability one.}
\end{figure}
\vspace{-8mm}

\begin{figure}[H]
\centering
\begin{subfigure}{.5\textwidth}
  \centering
  \includegraphics[width=.95\linewidth]{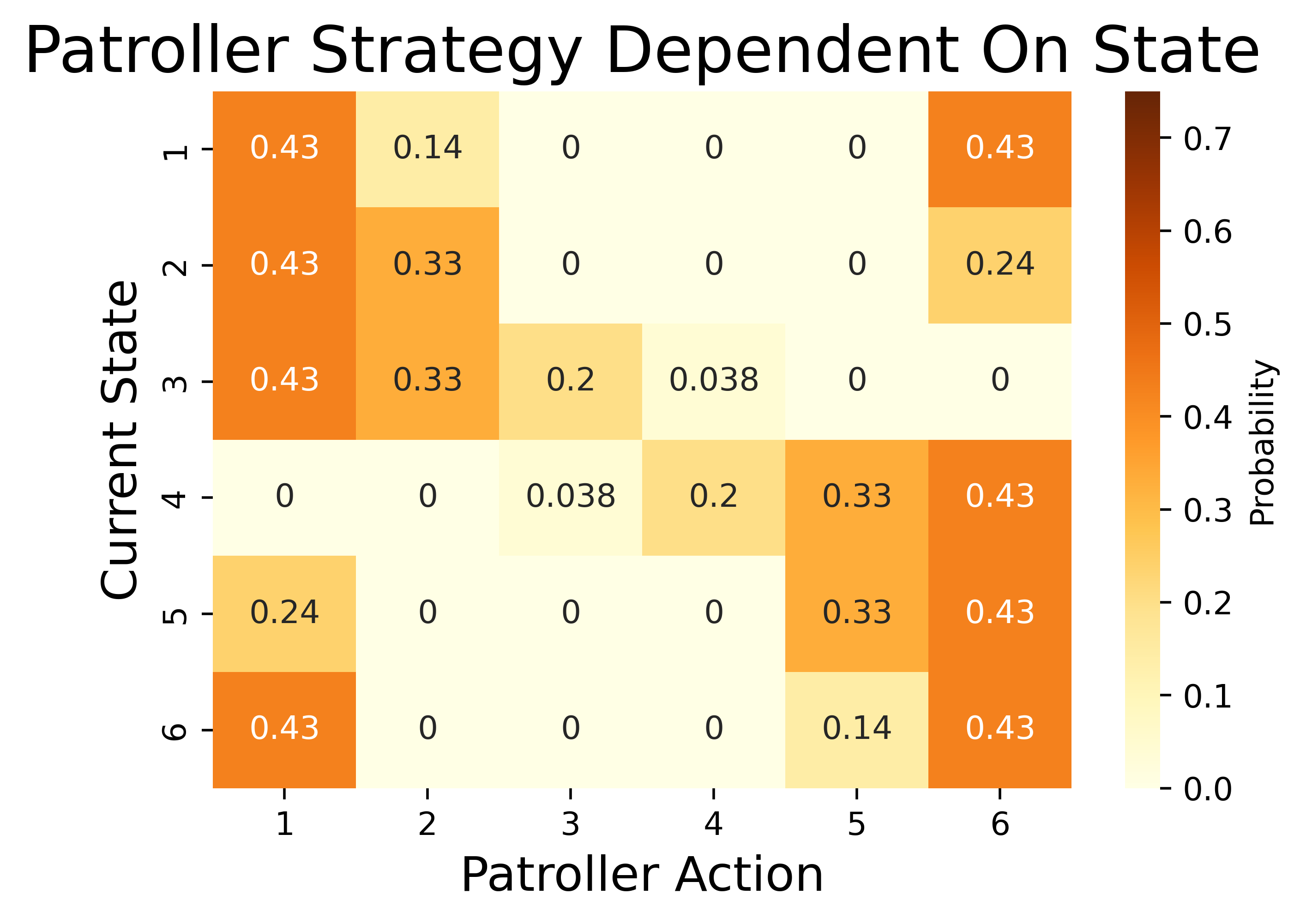}
  \caption{}
\end{subfigure}%
\begin{subfigure}{.5\textwidth}
  \centering
  \includegraphics[width=.95\linewidth]{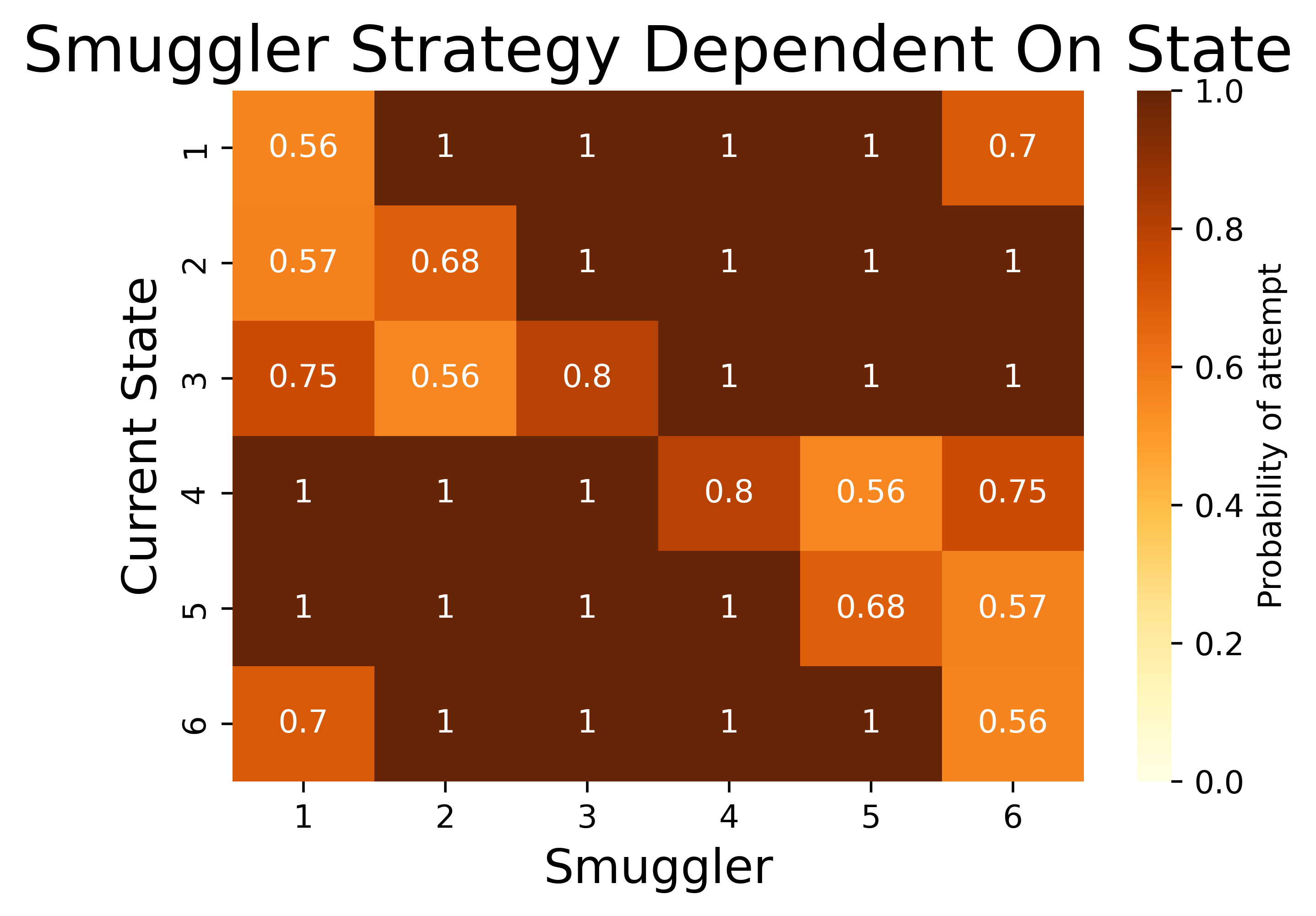}
  \caption{}
\end{subfigure}
\caption*{A Nash equilibrium in Example 3. The figure has the same interpretation as Figure 1.}
\end{figure}
\vspace{-8mm}

\begin{figure}[H]
\centering
\begin{subfigure}{.5\textwidth}
  \centering
  \includegraphics[width=.95\linewidth]{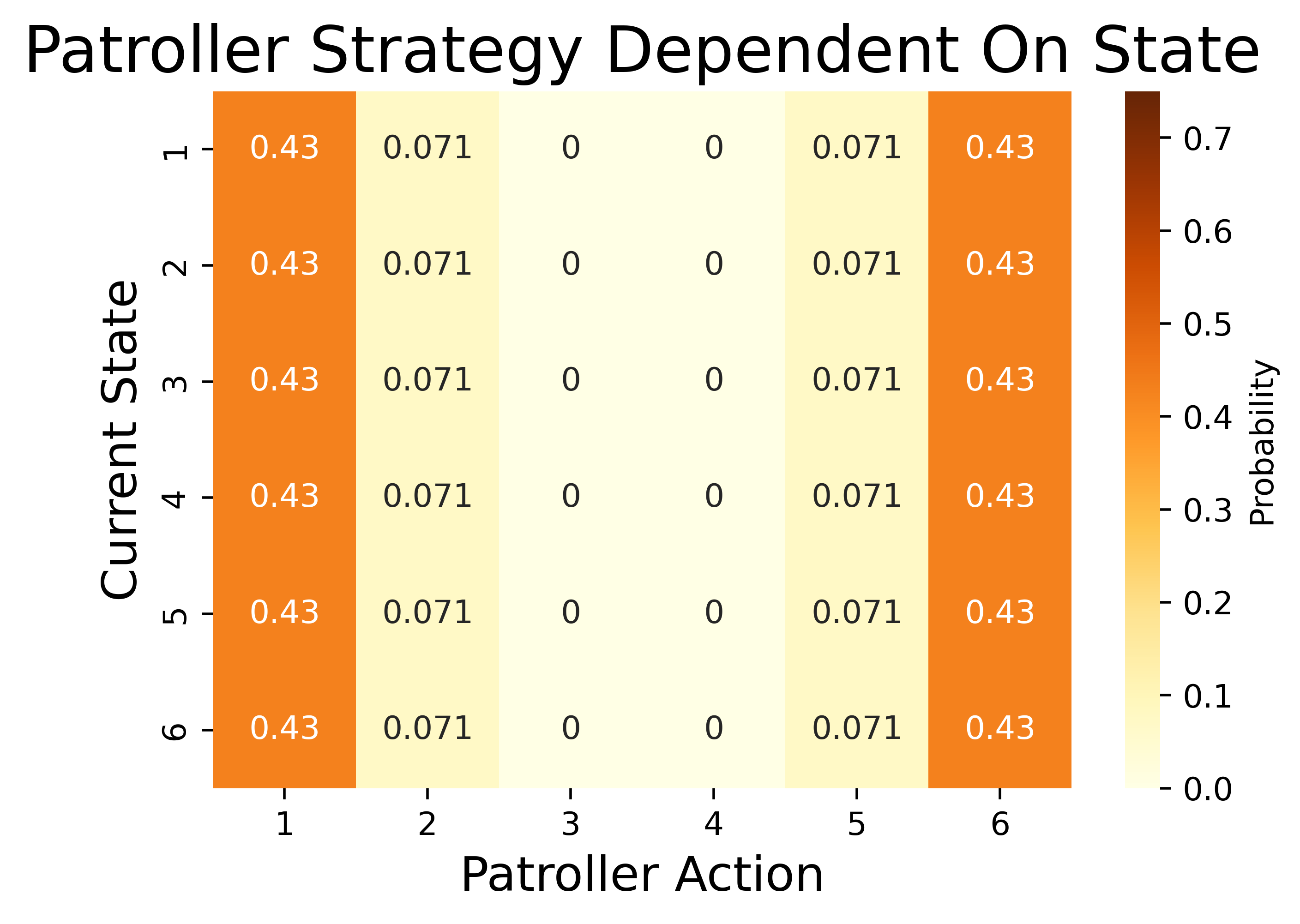}
  \caption{}
\end{subfigure}%
\begin{subfigure}{.5\textwidth}
  \centering
  \includegraphics[width=.95\linewidth]{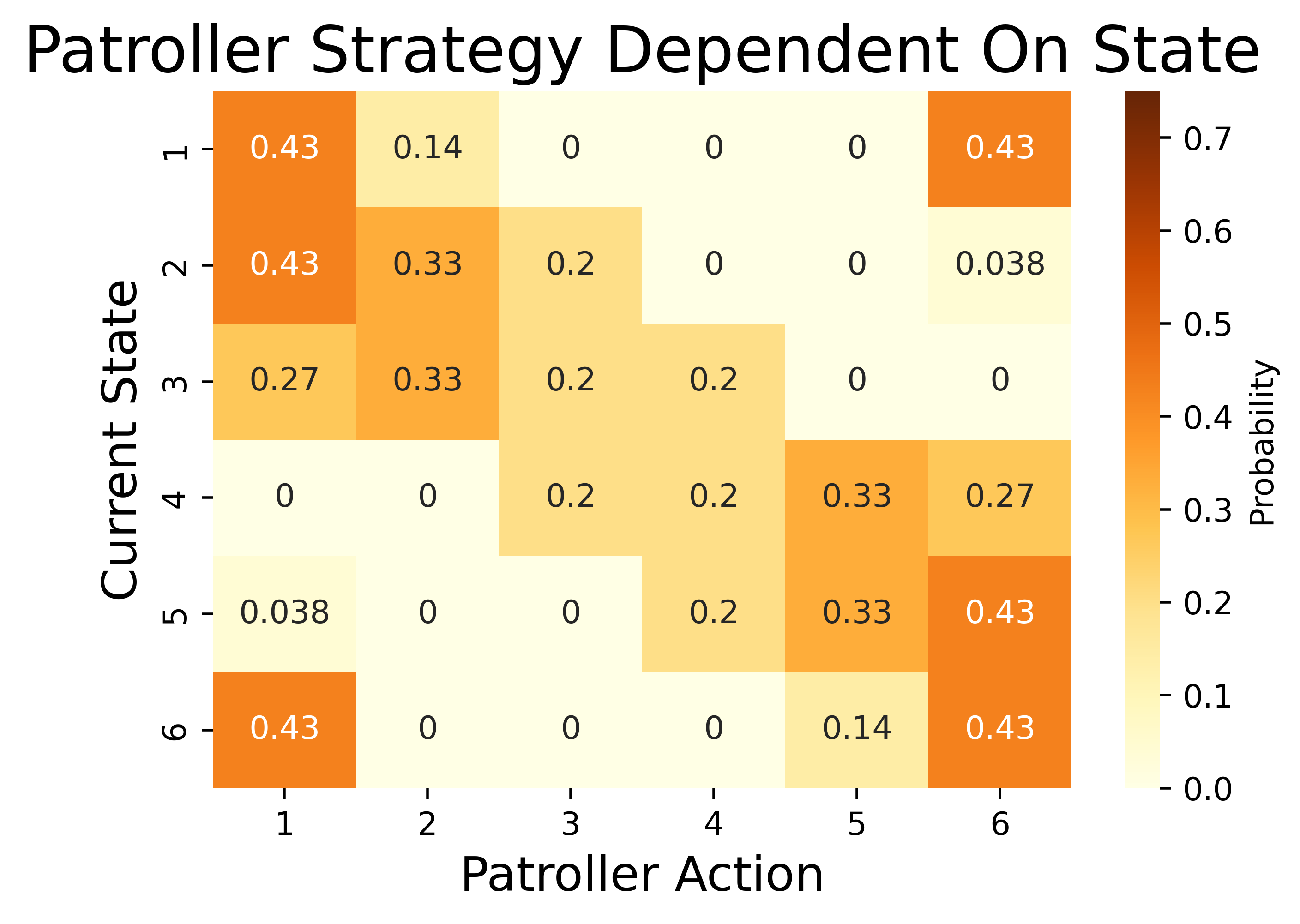}
  \caption{}
\end{subfigure}
\caption*{Patroller's strategies in a Nash equilibrium for the two described models. The figures gave the same interpretation as Figure 1(a).}
\end{figure}
\vspace{-8mm}

\newpage

\section{Appendix}

\subsection{Proof of Proposition 1}
\begin{proof}[Proof]
Assume that all players have a discount rate of $\gamma$ and that we have a Nash equilibrium of $(\boldsymbol{\Pi}^*, \boldsymbol{\Xi}^*)$. We first show that after changing a single smuggler's discount rate we have an unchanged Nash equilibrium. Then, by induction we can apply this to every smuggler in turn to see that with discount rates $(\gamma, \lambda_1, \dots, \lambda_n)$ we still have that $(\boldsymbol{\Pi}^*, \boldsymbol{\Xi}^*)$ is a Nash equilibrium.

Suppose that we alter smuggler $j$'s discount rate to $\lambda_j$. First, we consider each player other than smuggler $j$. Since their discount factor is still the same, their expected reward is still the same, and thus they will not have any incentive to deviate from their strategy. The only player who may have an incentive to change their strategy is the smuggler $j$. However, since the patroller controls the state transitions in the game, the smuggler $j$ can only maximise his instantaneous reward (further details can be seen in Corollary 1). A best response of smuggler $j$ does not depend on their discount factor, and therefore they also do not have an incentive to deviate. 
\end{proof}

\subsection{Proof of Proposition 2}
\begin{proof}[Proof]
    Firstly, we consider whether the patroller has an incentive to deviate from their strategy $\boldsymbol{\Pi}^*$. The patroller's reward is the same in both games under any strategy taken by either player. Therefore, the patroller not having an incentive to deviate in one game implies that they have no incentive to deviate in the other.
    
    Secondly, we explore whether the smuggler has any incentive to deviate from their strategy $\boldsymbol{\Xi}^*$. The patroller is the single controller in the stochastic game, and thus the smugglers can only try to maximize their instantaneous reward (see Corollary 1). The difference between the reward functions for the smugglers in the two games does not depend on their action $\boldsymbol{a}$ since,
    \begin{equation*}
        R_{smug}(b, \boldsymbol{a}) - \tilde{R}_{smug}(b, \boldsymbol{a}, s) = m_{s,b}.
    \end{equation*}
    Therefore, they have no incentive to deviate from their strategy $\boldsymbol{\Xi}^*$ in one game if and only if they have no incentive to deviate from $\boldsymbol{\Xi}^*$ in the other.
    
    Hence the Nash equilibria in the games coincide.
\end{proof}

\subsection{Proof of Proposition 3}
\begin{proof}[Proof]
The expected payoff to the patroller when the smugglers take a best response against them can be written as,
\begin{align}
    \min_{ \boldsymbol{a} \in [0,1]^n } \left\{ \sum_{b = 1}^n \pi_b [ \tilde{R}_{pat}(b,\boldsymbol{a},s) + \gamma \boldsymbol{V}_{pat}(b) ] \right\} = & \min_{ \boldsymbol{a} \in [0,1]^n } \left\{ \sum_{b = 1}^n \pi_b \tilde{R}_{pat}(b,\boldsymbol{a},s) + \gamma \sum_{b = 1}^n \pi_b \boldsymbol{V}_{pat}(b) \right\} \label{m}.
\end{align}
The first sum can be rewritten if we expand upon the equation for the smugglers' reward function as follows,
\begin{align*}
 \sum_{b = 1}^n \pi_b  \tilde{R}_{pat}(b,\boldsymbol{a},s) &= \sum_{b = 1}^n \pi_b  \left[ C(a_b) - \sum_{i \in [n]\setminus\{b\}} r_i a_i - m_{s,b} \right] \\
&=  \sum_{b = 1}^n [\pi_b C(a_b) + (\pi_b - 1)r_b a_b  - \pi_b m_{s,b} ] .
\end{align*}
Thus, if we consider this in Equation (\ref{m}) we derive,
\begin{align*}
  \min_{ \boldsymbol{a} \in [0,1]^n } \left\{ \sum_{b = 1}^n \pi_b [ \tilde{R}_{pat}(b,\boldsymbol{a},s) + \gamma \boldsymbol{V}_{pat}(b) ] \right\} &= \min_{ \boldsymbol{a} \in[0,1]^n } \left\{\sum_{b = 1}^n [\pi_b C(a_b) + (\pi_b - 1)r_b a_b ] - \pi_b m_{s,b} + \gamma \pi_b \boldsymbol{V}_{pat}(b) \right\} \\
  &= \sum_{b = 1}^n \left[ \min_{ a_b \in [0, 1] } \left\{ \pi_b C(a_b) + (\pi_b - 1)r_b a_b  \right\} - \pi_b m_{s,b} + \gamma \pi_b \boldsymbol{V}_{pat}(b) \right] \\
  &= \sum_{b=1}^n \left[ - \max_{a_b \in [0, 1]} \left\{ (1 - \pi_b) r_b a_b - \pi_b C(a_b) \right\} + \pi_b (\gamma \boldsymbol{V}_{pat}(b) - m_{s,b}) \right]
\end{align*}
as required. 
\end{proof}

\subsection{Proof of Lipschitz Constant in Lemma 2}
\begin{proof}[Proof]
Let $\delta > 0$, then the Lipschitz constant $L$ can be taken as,
\begin{equation}
    L \leq \max \left\{ \left| \frac{g_b(1+2\delta) - g_b(1+\delta)}{(1+2\delta) - (1+\delta)} \right|, \left| \frac{g_b(-2\delta) - g_b(-\delta)}{(-2\delta) - (-\delta)} \right| \right\} \label{L}.
\end{equation}
We can see that,
\begin{align*}
    & \argmax_{a \in [0,1]} \left\{ (1 - \pi_b) r_b a - \pi_b C(a) \right\} = \begin{dcases} 1 \text{ if } \pi_b < 0 \\ 0 \text{ if } \pi_b > 1 \end{dcases} \\
    \implies & \max_{a \in [0,1]} \left\{ (1 - \pi_b) r_b a - \pi_b C(a) \right\} = \begin{dcases} (1 - \pi_b)r_b - \pi_b C(1) &\text{ if } \pi_b < 0 \\ 0 &\text{ if } \pi_b > 1 \end{dcases} 
\end{align*}
and so then we can evaluate (\ref{L}) since,
\begin{equation*}
    \left| \frac{g_b(1+2\delta) - g_b(1+\delta)}{(1+2\delta) - (1+\delta)} \right| = \frac{1}{\delta} \left| \delta(\gamma \boldsymbol{V}_{pat}(b) - m_{s,b}) \right| = \left| \gamma \boldsymbol{V}_{pat}(b) - m_{s,b} \right|
\end{equation*}
and,
\begin{equation*}
    \left| \frac{g_b(-2\delta) - g_b(-\delta)}{(-2\delta) - (-\delta)} \right| = \frac{1}{\delta} \left| - \delta C(1) - \delta r_b - \delta (\gamma \boldsymbol{V}_{pat}(b) - m_{s,b}) \right| = \left| r_b + C(1) + \gamma \boldsymbol{V}_{pat}(b) - m_{s,b} \right|.
\end{equation*}
From the above calculations and (\ref{L}) we conclude that,
\begin{align*}
    L &= \max \left\{ \Big| r_b + C(1) + \gamma \boldsymbol{V}_{pat}(b) - m_{s,b} \Big| ,  \Big| \gamma \boldsymbol{V}_{pat}(b) - m_{s,b} \Big| \right\} \\
    & \leq r_b + C(1) - (\gamma \boldsymbol{V}_{pat}(b) - m_{s,b}) 
\end{align*}
since $r_b$ and $C(1)$ are positive but $\gamma \boldsymbol{V}_{pat}(b)$ and $-m_{s,b}$ are negative. This conludes the proof. 
\end{proof}

\subsection{Proof of Theorem 1}
\begin{proof}[Proof]
Recall that $\boldsymbol{\pi}^*$ is the optimal solution,  $\tilde{\boldsymbol{\pi}}$ is found using Algorithm 1 and $\hat{\boldsymbol{\pi}}$ is constructed using Algorithm 3. Denote the value of $\hat{\boldsymbol{\pi}}$ after iteration $l$ of Algorithm 3 by $\hat{\boldsymbol{\pi}}(l)$. We let,
\begin{equation*}
    K_m = m \prod_{b=1}^n [C(1) + r_b],
\end{equation*} and denote the output of Algorithm 1 when using $K = K_m$ as $\tilde{\boldsymbol{\pi}}_{K_m}$. Furthermore, we respresent its value after a step of Algorithm 1 with value $k$ by $\tilde{\boldsymbol{\pi}}_{K_m}(k)$. Since $K_t$ is divisible by every $C(1) + r_b$, it follows that for every $l$ there exists a $k$ such that,
\begin{equation*}
    \sum_{b=1}^n \hat{\pi}_b(l) = \sum_{b=1}^n \tilde{\pi}_{K_m,b}(k)
\end{equation*}
and we denote this $k$ by $k_l$. 

We prove by induction that for every $l$ we have $\hat{\boldsymbol{\pi}}(l) = \tilde{\boldsymbol{\pi}}_{K_m}(k_l)$. This is clearly true when $l = 0$ since $ \hat{\boldsymbol{\pi}}(0) = \tilde{\boldsymbol{\pi}}_{K_m}(k_0) = \boldsymbol{0}$. Assume that for a given $l$ we have $\hat{\boldsymbol{\pi}}(l) = \tilde{\boldsymbol{\pi}}_{K_m}(k_l)$. In Step 3 of Algorithm 3 we find,
\begin{equation*}
    j \in \argmax_{b \in [n]} \left\{ \frac{g_b(\hat{\pi}_b(l) + x_b) - g_b(\hat{\pi}_b(l))}{x_b} \right\} .
\end{equation*}
Since $g_b$ is linear on the interval $\pi_b \in [\hat{\pi}_b(l), \hat{\pi}_b(l) + x_b - \frac{1}{K}]$, this implies that,
\begin{equation*}
    j \in \argmax_{b \in [n]} \left\{ g_b \left( \pi_b+\frac{1}{K} \right) - g_b(\pi_b) \right\}
\end{equation*}
for $\pi_b \in [\hat{\pi}_b(l), \hat{\pi}_b(l) + x_b - \frac{1}{K}]$. Thus, at every step between $k_l$ and $k_{l+1}$ in Algorithm 1 we increase $\tilde{\pi}^j$ and so,
\begin{equation*}
    \tilde{\pi}_{K_m, b}(k_{l+1}) = \tilde{\pi}_{K_m, b}(k_l) + \mathbbm{1}(b = j) x_b = \hat{\pi}_b(k_l) + \mathbbm{1}(b = j) x_b = \hat{\pi}_b(l+1) .
\end{equation*}
Thus, we have $\tilde{\boldsymbol{\pi}}_{K_m}(k_{l+1}) = \hat{\boldsymbol{\pi}}(l+1)$, so by induction we deduce that $\tilde{\boldsymbol{\pi}}_{K_m} = \hat{\boldsymbol{\pi}}$. For each $m \in \mathbb{N}$ the proximity result by \cite{Hochbaum1994} implies that,
\begin{equation*}
     \| \boldsymbol{\pi}^* - \tilde{\boldsymbol{\pi}}_{K_m} \|_{\infty} \leq \frac{n}{K_m}
\end{equation*}
which therefore means that since $K_m \to \infty$ as $m \to \infty$ we must have $\| \boldsymbol{\pi}^* - \tilde{\boldsymbol{\pi}}_{K_m} \|_{\infty} \to 0$. Since we know that $\| \boldsymbol{\pi}^* - \tilde{\boldsymbol{\pi}}_{K_m} \|_{\infty} = \| \boldsymbol{\pi}^* - \hat{\boldsymbol{\pi}} \|_{\infty}$ for all $m \in \mathbb{N}$ we must have that $\boldsymbol{\pi}^* = \hat{\boldsymbol{\pi}}$.
\end{proof}

\end{document}